\documentclass[runningheads]{llncs}
\usepackage[T1]{fontenc}
\usepackage{enumitem}
\usepackage{hyperref}
\usepackage{tikz}
\usepackage{amsmath}
\usepackage{graphicx}
\usepackage{tikz}
\usetikzlibrary{positioning,arrows.meta}

\usepackage{graphicx}
\begin{document}
\title{The Meta-rotation Poset for Student-Project Allocation}
%
%
\author{Peace Ayegba\inst{1}\orcidID{0000-0002-0830-7811} \and \\
Sofiat Olaosebikan\inst{1}\orcidID{0000-0001-6754-7308}}
%
%
\institute{School of Computing Science, University of Glasgow}
\maketitle              
\begin{abstract}
We study the Student–Project Allocation problem with lecturer preferences over Students (\textsc{spa-s}), an extension of the well-known Stable Marriage and Hospital–Residents problem. In this model, students have preferences over projects, each project is offered by a single lecturer, and lecturers have preferences over students. The goal is to compute a \textit{stable matching}, which is an assignment of students to projects (and thus to lecturers) such that no student or lecturer has an incentive to deviate from their current assignment. While motivated by the university setting, this problem arises in many allocation settings where limited resources are offered by agents with their own preferences, such as in wireless networks.

We establish new structural results for the set of stable matchings in \textsc{spa-s} by developing the theory of \emph{meta-rotations}, a generalisation of the well-known notion of rotations from the Stable Marriage problem. Each meta-rotation corresponds to a minimal set of changes that transforms one stable matching into another within the lattice of stable matchings. The set of meta-rotations, ordered by their precedence relations, forms the \emph{meta-rotation poset}. We prove that there is a one-to-one correspondence between the set of stable matchings and the closed subsets of the meta-rotation poset. By developing this structure, we provide a foundation for the design of efficient algorithms for enumerating and counting stable matchings, and for computing other optimal stable matchings, such as egalitarian or minimum-cost matchings, which have not been previously studied in {\sc spa-s}

\end{abstract}

\section{Introduction}
Matching problems occur in settings where one set of agents must be assigned to another subject to capacity constraints and/or preferences. Since the introduction of the Stable Marriage problem ({\sc sm}) and the seminal Gale–Shapley algorithm \cite{gale1962college,mcvitie1971stable}, matching problems have been studied extensively from both theoretical and practical perspectives \cite{iwama2008survey,knuth1997stable,gusfield1987three,manlove2013algorithmics}. The Student–Project Allocation problem with lecturer preferences over Students (\textsc{spa-s}) extends classical stable matching models. In this problem, students express preferences over available projects, each offered by a lecturer, and lecturers express preferences over the students. Each project and lecturer has a capacity constraint, and a matching assigns students to projects so that neither project nor lecturer capacities are exceeded. A matching is said to be \emph{stable} if there is no student and lecturer who would both prefer to be matched together than with their current assignments.  

Abraham et al.~\cite{AIM2007} showed that every instance of {\sc spa-s} admits at least one stable matching and presented two polynomial-time algorithms to find such matchings. The \emph{student-oriented} algorithm produces the student-optimal stable matching, where every student obtains their best possible project among all stable matchings, while the \emph{lecturer-oriented} algorithm yields the lecturer-optimal stable matching, where each lecturer receives their best set of students. Moreover, a single instance may admit several stable matchings other than these two matchings. The authors ~also proved properties satisfied by all stable matchings in a given instance, known as the \emph{Unpopular Projects Theorem}, which we state in what follows:

\begin{theorem}[\cite{AIM2007}]
\label{thm:unpopular-students}
In any {\sc spa-s} instance:
\begin{enumerate}[label=(\roman*)]
    \item the same students are assigned in all stable matchings;
    \item each lecturer is assigned the same number of students; and
    \item if a project is offered by an undersubscribed lecturer, it receives the same number of students in all stable matchings.
\end{enumerate}
\end{theorem}

We remark that \textsc{spa-s} generalises the Hospital-Residents problem (\textsc{hr}) \cite{manlove2015hospitals}, where projects and lecturers are effectively indistinguishable. In \textsc{hr} setting, lecturers (and projects) correspond to hospitals, while students correspond to residents. Moreover, the set of stable matchings in this model satisfies well-defined structural properties, collectively referred to as the \textit{Rural Hospitals Theorem}. However, not all of its properties extend to {\sc spa-s}; for example, an undersubscribed lecturer in {\sc spa-s} may be assigned different students in different stable matchings, whereas an undersubscribed hospital in {\sc hr} is assigned the same set of residents across all stable matchings. 

A central line of research on stable matchings studies how the set of all stable matchings forms a distributive lattice, how the corresponding Hasse diagram can be generated, and how this structure can be traversed efficiently \cite{ayegba2025structural,blair1988lattice,GI1989,fleiner2003fixed,gangam2023stable,manlove2002structure}. Further, existing work has shown how these structures can be exploited to design efficient algorithms for various optimisation tasks \cite{boehmer2025adapting,GI1989,irving1987efficient,karzanov2025stable,eirinakis2012finding}. In the classical {\sc sm} problem, Gusfield and Irving~\cite{GI1989} introduced the \emph{rotation poset}, a compact representation of the structure of all stable matchings in a given instance. Although the number of stable matchings in an instance may be exponential in the size of the input, the rotation poset can be constructed in polynomial time. Moreover, this rotation poset allows us to derive one stable matching from another stable matching. Bansal~\cite{BAM2007} extended this idea to the many-to-many setting through the concept of \emph{meta-rotations}, and Cheng \cite{CMS2008} further adapted it to the  {\sc hr} problem, providing an algorithm to identify all meta-rotations in a given instance instance and using this to develop efficient algorithms for computing optimal stable matchings with respect to different objective functions.

We note that existing definitions and proofs of meta-rotations in the {\sc hr} setting do not directly carry over to the {\sc spa-s} setting due to the presence of projects. In the {\sc hr} setting~\cite{cheng2023stable}, the definition of a meta-rotation relies on the observation that when a hospital \(h\) becomes better or worse off, its least-preferred resident must change. However, this property does not hold in {\sc spa-s}: a lecturer may be better off in one matching compared to another while its least-preferred student remains the same (although some student assigned to the lecturer must change). This observation, among others, motivates the need for a refined definition of meta-rotations that is specific to the {\sc spa-s} setting. Subsequent research also extended the notion of rotations to setting with ties and incomplete preference lists \cite{mcdermid2011structural,cheng2023stable}. Scott~\cite{S2005} defined meta-rotations for super-stable matchings in {\sc smti}, proving that there exists a one-to-one correspondence between the set of super-stable matchings and the family of closed subsets of the meta-rotation poset. Hu and Garg~\cite{HG2021} later gave an alternative construction of this representation in $O(mn)$ time.

\medskip
\textbf{Our Contributions.} 
We develop the theory of \emph{meta-rotations} for {\sc spa-s}, extending the classical notion of rotations from {\sc sm} and establishing analogous structural results that have not previously been derived for this setting. We formally define meta-rotations and show that each represents a minimal set of changes transforming one stable matching into another. We further define the \emph{meta-rotation poset}, a partial order capturing the dependencies among meta-rotations and providing a compact representation of all stable matchings in an instance. We then prove a one-to-one correspondence between the set of stable matchings and the family of closed subsets of the meta-rotation poset.
This correspondence, implied by Birkhoff’s Representation Theorem \cite{birkhoff1937rings}, yields a constructive way to generate all stable matchings and to identify other optimal or desirable stable matchings beyond the student- and lecturer-optimal ones.

\section*{2\;\;Preliminaries}
In the Student--Project Allocation problem with lecturer preferences over Students (\textsc{spa-s}), we have a set of students $\mathcal{S} = \{s_1, \ldots, s_{n_1}\}$, a set of projects $\mathcal{P} = \{p_1, \ldots, p_{n_2}\}$, and a set of lecturers $\mathcal{L} = \{l_1, \ldots, l_{n_3}\}$. 
Each project is offered by exactly one lecturer, and each lecturer~$l_k$ offers a non-empty subset $P_k \subseteq \mathcal{P}$ of projects, with the sets $P_1, \ldots, P_{n_3}$ forming a partition of~$\mathcal{P}$. Each student~$s_i$ provides a strict preference ordering over a subset of projects that they find acceptable. 
Each lecturer~$l_k$ also has a strict preference ordering over the students who find at least one project in~$P_k$ acceptable.

A pair $(s_i, p_j)$, where ~$p_j$ is offered by~$l_k$, is called \emph{acceptable} if $p_j$ appears on $s_i$'s preference list and $s_i$ appears on $l_k$'s list. Each project~$p_j$ has a capacity $c_j$, and each lecturer~$l_k$ has a capacity $d_k$. 
An \emph{assignment} $M$ is a set of acceptable student--project pairs. 
We write $M(s_i)$ to denote the project assigned to~$s_i$, if any, and $M(p_j)$ and $M(l_k)$ for the sets of students assigned to ~$p_j$ and ~$l_k$, respectively. 
A \emph{matching} is an assignment~$M$ such that $|M(s_i)| \le 1$ for every $s_i \in \mathcal{S}$, $|M(p_j)| \le c_j$ for every $p_j \in \mathcal{P}$, and $|M(l_k)| \le d_k$ for every $l_k \in \mathcal{L}$.

\begin{definition}[Stability in \textsc{spa-s}]
\label{def:stability}
Let $I$ be an instance of \textsc{spa-s} and $M$ a matching in $I$. An acceptable pair $(s_i, p_j) \notin M$, where $p_j$ is offered by lecturer~$l_k$, is a \emph{blocking pair} in~$M$ 
if $s_i$ is unassigned in~$M$ or prefers $p_j$ to $M(s_i)$, and one of the following holds:
\begin{enumerate}[label=(\alph*)]
    \item both $p_j$ and $l_k$ are undersubscribed in $M$;
    \item $p_j$ is undersubscribed in $M$, $l_k$ is full in $M$, and either $s_i \in M(l_k)$ or $l_k$ prefers $s_i$ to the worst student in $M(l_k)$;
    \item $p_j$ is full and $l_k$ prefers $s_i$ to the worst student in $M(p_j)$.
\end{enumerate}
A matching is \emph{stable} if it admits no blocking pair.
\end{definition}

\begin{definition}[Student preferences over matchings]
Let \( \mathcal{M} \) denote the set of all stable matchings in a {\sc spa-s} instance \( I \). Given two matchings \( M, M' \in \mathcal{M} \), a student \( s_i \in \mathcal{S} \) \emph{prefers} \( M \) to \( M' \) if \( s_i \) is assigned in both matchings and prefers \( M(s_i) \) to \( M'(s_i) \). Similarly, \( s_i \) is \emph{indifferent} between \( M \) and \( M' \) if either \( s_i \) is unassigned in both \( M \) and \( M' \), or \( M(s_i) = M'(s_i) \).
\end{definition}

\begin{definition}[Lecturer preferences over matchings]
\label{sect:lect-pref}
\noindent Let $M$ and $M'$ be two stable matchings in \( \mathcal{M} \). We recall from Theorem \ref{thm:unpopular-students} that \( |M| = |M'| \) and \( |M(l_k)| = |M'(l_k)| \) for each lecturer \( l_k \). Suppose that \( l_k \) is assigned different sets of students in \( M \) and \( M' \). We define
\(
M(l_k) \setminus M'(l_k) = \{s_1, \ldots, s_r\}\) and 
\(M'(l_k) \setminus M(l_k) = \{s'_1, \ldots, s'_r\}\), where the students in each set are listed in the order they appear in \( l_k \)'s preference list \( \mathcal{L}_k \). Then \( l_k \) \emph{prefers} \( M \) to \( M' \) if \( l_k \) prefers \( s_i \) to \( s'_i \) for all \( i \in \{1, \ldots, r\} \).
\end{definition}

\begin{definition}[{Dominance relation}]
Let \( M, M' \in \mathcal{M} \). We say that \( M \) \emph{dominates} \( M' \), denoted \( M \preceq M' \), if and only if each student prefers \( M \) to \( M' \), or is indifferent between them.
\end{definition}

\noindent From this definition, we observe that if a lecturer \( l \) is assigned different sets of students in two stable matchings \( M \) and \( M' \), they do not necessarily prefer each student in \( M(l) \) to those in \( M'(l) \setminus M(l) \), nor each student in \( M'(l) \) to those in \( M(l) \setminus M'(l) \). However, it is always the case that \( l \) prefers at least one student in \( M(l) \setminus M'(l) \) to at least one student in \( M'(l) \setminus M(l) \), or vice versa. This contrasts with the {\sc hr} setting, where given any two stable matchings \( M \) and \( M' \), each hospital either prefers all of its assigned residents in \( M \) to those in \( M' \setminus M \), or all its assigned residents in \( M' \) to those in \( M \setminus M' \).

\medskip
\noindent \textbf{Example:} Consider the {\sc spa-s} instance $I$ in Figure \ref{figure:leadstorotation}. There are two stable matchings in $I$ namely \( M_1 = \{(s_1, p_1), (s_2, p_3), \) \( (s_3, p_2), (s_4, p_4)\} \), and \( M_2 = \{(s_1, p_2), (s_2, p_4), \) \( (s_3, p_1), (s_4, p_3)\} \). Each student prefers their assigned project in \( M_1 \) to that in \( M_2 \); hence \( M_1 \) dominates \( M_2 \).

\begin{figure}[h]
\centering
\small
\begin{tabular}{llll}
\hline
Students' preferences & Lecturers' preferences  &  offers \\ 
$s_1$: \;$p_1$ \; $p_2$ \; &  $l_1$: \; $s_1$ \; $s_3$ \; & $p_2$ \\ 
$s_2$: \;$p_3$ \; $p_4$ \; &  $l_2$: \; $s_2$ \; $s_4$ \; & $p_4$ \\
$s_3$: \;$p_2$ \;  $p_1$ \; & $l_3$: \; $s_3$ \;$s_4$ \;$s_1$ & $p_1$ \\ 
$s_4$: \;$p_4$ \;  $p_1$ \; $p_3$ \; & $l_4$: \; $s_4$ \;$s_2$\;$s_1$  & $p_3$ \\ 

&  \\
& Project capacities: $\forall c_j = 1$ && \\
&  Lecturer capacities: $\forall d_k = 1$& &\\ 
\hline
\end{tabular}
\caption{ An instance $I$ of {\sc spa-s}}
\label{figure:leadstorotation}
\end{figure}

\section{Structural results involving stable matchings}
\label{sect:struct-results-stablematchings}
In this section, we present new results on stable matchings in a {\sc spa-s} instance, providing insight into how the assignment of a student to different projects in two stable matchings affects the preferences of the involved lecturers. Throughout, let \( l_k \) denote the lecturer offering project \( p_j \). 

\begin{lemma}
\label{lem:swaps}
Let \( M \) and \( M' \) be stable matchings such that \( M \) dominates \( M' \). If a student \( s_i \) is assigned to different projects in \( M \) and \( M' \), with \( s_i \) assigned to \( p_j \) in \( M' \), then:
\begin{enumerate}[label=(\roman*)]
    \item if \( p_j \) is full in \( M \), the worst student in \( M(p_j) \) is not in \( M'(p_j) \);
    \item if \( p_j \) is undersubscribed in \( M \), the worst student in \( M(l_k) \) is not is not in \( M'(l_k) \).
\end{enumerate}
\end{lemma}

\begin{proof}
Suppose that \( s_i \) is some student assigned to different projects in \( M \) and \( M' \), such that \( s_i \in M'(p_j) \setminus M(p_j) \). Let \(s_z\) be the worst student in \(M(p_j)\), and suppose for a contradiction that $s_z \in M(p_j) \cap M'(p_j)$. Consider case (i) where \(p_j\) is full in \(M\). Since \(s_i \in M'(p_j) \setminus M(p_j)\) and \( |M(p_j)| \geq |M'(p_j)|\), there exists some student \(s_t \in M(p_j) \setminus M'(p_j)\). Moreover, since \(s_z\) is the worst student in \(M(p_j)\), $l_k$ prefers \(s_t\) to \(s_z\). Since \(M\) dominates \(M'\), \(s_t\) prefers \(M\) to \(M'\). Regardless of whether \( p_j \) is full or undersubscribed in \( M' \), the pair \( (s_t, p_j) \) blocks \( M' \), a contradiction. Therefore, case (i) holds.

\vspace{0.2cm}
Now consider case (ii) where $p_j$ is undersubscribed in $M$. Let $s_z$ be the worst student in $M(l_k)$, and suppose for a contradiction that $s_z \in M(l_k) \cap M'(l_k)$. First, suppose that \( |M(p_j)| \geq |M'(p_j)| \). Since \( p_j \) is  undersubscribed in \( M \), it follows that \( p_j \) is undersubscribed in \( M' \). Given that \( s_i \in M'(p_j) \setminus M(p_j) \), there exists some student \( s_r \in M(p_j) \setminus M'(p_j) \). Furthermore, \( s_r \) prefers \( M \) to \( M' \), and either \( s_r = s_z \) or \( l_k \) prefers \( s_r \) to \( s_z \). If \( s_r = s_z \), then \( s_r \in M'(l_k) \) and, since \( p_j \) is undersubscribed in \( M' \), the pair \( (s_r, p_j) \) blocks \( M' \), leading to a contradiction. If instead \( s_r \neq s_z \), then \( l_k \) prefers \( s_r \) to \( s_z \), since \( s_z \) is the worst student in \( M(l_k) \). However, given that \( s_r \) prefers \( M \) to \( M' \), \( p_j \) is undersubscribed in \( M' \), and \( l_k \) prefers \( s_r \) to \( s_z \), the pair \( (s_r, p_j) \) blocks \( M' \), again leading to a contradiction.  

\vspace{0.2cm}
Suppose that \( |M'(p_j)| > |M(p_j)| \). Since \( |M(l_k)| = |M'(l_k)| \), there exists some project \( p_t \in P_k \) such that \( |M(p_t)| > |M'(p_t)| \), meaning \( p_t \) is undersubscribed in \( M' \). Consequently, there exists a student \( s_t \in M(p_t) \setminus M'(p_t) \) who prefers \( M \) to \( M' \). If \( s_t = s_z \), then \( s_t \in M'(l_k) \) and, since \( p_t \) is undersubscribed in \( M' \), then \( (s_t, p_t) \) blocks \( M' \), leading to a contradiction. Otherwise, since \( s_z \) is the worst student in \( M(l_k) \), it follows that \( l_k \) prefers \( s_t \) to \( s_z \). Given that \( s_t \) prefers \( M \) to \( M' \), \( p_t \) is undersubscribed in \( M' \), and \( l_k \) prefers \( s_t \) to \( s_z \), the pair \( (s_t, p_t) \) blocks \( M' \), a contradiction. Hence, our claim holds.
\end{proof}

\begin{lemma}
\label{lem:s-prefers-M-l-prefers-M':dominance}
Let \( M \) and \( M' \) be stable matchings in an instance \( I \) such that \( M \) dominates \( M' \). 
If a student \( s_i \) is assigned to different projects in \( M \) and \( M' \), with \( s_i \) assigned to \( p_j \) in \( M' \), then:
\begin{enumerate}[label=(\roman*)]
    \item if \( p_j \) is full in \( M \), \( l_k \) prefers \( s_i \) to the worst student in \( M(p_j) \); 
    \item if \( p_j \) is undersubscribed in \( M \), \( l_k \) prefers \( s_i \) to the worst student in \( M(l_k) \).
\end{enumerate}
\end{lemma}

\begin{proof}
\noindent Let \( M \) and \( M' \) be two stable matchings in \( I \), where \( M \) dominates \( M' \). Suppose  \( s_i \) is assigned to \( p_j \) in \( M' \), where \( l_k \) offers \( p_j \) (and possibly \( l_k \) offers \( M(s_i) \)). Consider case (i), where \( p_j \) is full in \( M \). Let \( s_z \) be the worst student in \( M(p_j) \), and suppose for a contradiction that \( l_k \) prefers \( s_z \) to \( s_i \). By Lemma \ref{lem:swaps}, it follows that \( s_z \in M(p_j) \setminus M'(p_j) \). Since \( M \) dominates \( M' \), \( s_z \) prefers \( p_j \) to \( M'(s_z) \).  If \( p_j \) is full in \( M' \), then the pair \( (s_z, p_j) \) blocks \( M' \), since \( l_k \) prefers \( s_z \) to some student in \( M'(p_j) \), namely \( s_i \). Similarly, if \( p_j \) is undersubscribed in \( M' \), \( (s_z, p_j) \) also blocks \( M' \), since \( l_k \) prefers \( s_z \) to some student in \( M'(l_k) \), namely \( s_i \). This leads to a contradiction. Hence, $l_k$ prefers $s_i$ to $s_z$, and case (i) holds.

\vspace{0.2cm}
Consider case (ii), where \( p_j \) is undersubscribed in \( M \). Suppose for a contradiction that \( l_k \) prefers the worst student in \( M(l_k) \) to \( s_i \). First, suppose that \( |M(p_j)| \geq |M'(p_j)| \). Then, \( p_j \) is undersubscribed in \( M' \). Since \( M(p_j) \) contains at least as many students as \( M'(p_j) \), there exists some student \( s_r \in M(p_j) \setminus M'(p_j) \) (Readers may recall that \( s_i \in M'(p_j) \setminus M(p_j) \)). Additionally, \( s_r \) prefers \( M \) to \( M' \), since $M$ dominates $M'$. Given that \( s_r \in M(l_k) \) and \( s_r \) is either the worst student in \( M(l_k) \) or better, it follows that \( l_k \) prefers \( s_r \) to \( s_i \).
However, since \( p_j \) is undersubscribed in \( M' \) and \( l_k \) prefers \( s_r \) to some student in \( M'(l_k) \) (namely \( s_i \)), the pair \( (s_r, p_j) \) blocks \( M' \), leading to a contradiction. 

\vspace{0.2cm}
Suppose instead that \( |M(p_j)| < |M'(p_j)| \). Since $|M(l_k)| = |M'(l_k)|$, there exists some other project \( p_t \in P_k \) such that \( |M'(p_t)| < |M(p_t)| \). This means \( p_t \) is undersubscribed in \( M' \) and there exists some student \( s_t \in M(p_t) \setminus M'(p_t)\), that is, $s_t \in M(l_k)$. Moreover, \( s_t \) prefers \( M \) to \( M' \). Since \( p_t \) is undersubscribed in \( M' \) and \( l_k \) prefers \( s_t \) to some student in \( M'(l_k) \) (namely \( s_i \)), the pair \( (s_t, p_t) \) blocks \( M' \), contradicting the stability of \( M' \). Thus, we reach a contradiction in both scenarios, completing the proof for case (ii).
\end{proof}

\begin{lemma}
\label{lem:p*-l*-undersubscribed}
Let $M$ and $M'$ be two stable matchings where $M$ dominates $M'$. Suppose that a student \( s_i \) is assigned to different projects in \( M \) and \( M' \), with \( s_i \) assigned to \( p_j \) in \( M' \). If $p_j$ is undersubscribed in $M$ then $l_k$ is full in $M$.
\end{lemma}
\begin{proof}
Let \( M \) and \( M' \) be two stable matchings where \( M \) dominates \( M' \). Suppose \( s_i \) is some student assigned to different projects in \( M \) and \( M' \), such that \( s_i \) is assigned to \( p_j \) in \( M' \), and \( l_k \) offers \( p_j \) (possibly \( l_k \) also offers \( M(s_i) \)). Now, suppose for a contradiction that both \( p_j \) and \( l_k \) are undersubscribed in \( M \). Since \( p_j \) is offered by an undersubscribed lecturer \( l_k \), it follows from Theorem \ref{thm:unpopular-students} that the same number of students are assigned to \( p_j \) in \( M \) and \( M' \). Therefore, since \( s_i \in M'(p_j) \setminus M(p_j) \), there exists some student \( s_z \) such that \( s_z \in M(p_j) \setminus M'(p_j) \). Moreover, both \( p_j \) and $l_k$ are undersubscribed in \( M' \), since $|M(p_j)| = |M'(p_j)|$ and $|M(l_k)| = |M'(l_k)|$. Since \( M \) dominates \( M' \), \( s_z \) prefers \( p_j \) to \( M'(s_z) \). However, since \( p_j \) and \( l_k \) are both undersubscribed in \( M' \), \( (s_z, p_j) \) blocks \( M' \), a contradiction. Hence, our claim holds.
\end{proof}

\noindent Finally, we recall existing results established in \cite{ayegba2025structural}, which provide additional insight into the behaviour of students assigned to different projects across stable matchings; these results are used in the subsequent proofs.
\begin{lemma}
\label{lem:chap4samelecturer1}
Let \( M \) and \( M' \) be two stable matchings in \( I \). If a student \( s_i \) is assigned in \( M \) and \( M' \) to different projects offered by the same lecturer \( l_k \), and \( s_i \) prefers \( M \) to \( M' \), then there exists some student \( s_r \in M'(l_k) \setminus M(l_k) \) such that \( l_k \) prefers \( s_r \) to \( s_i \). Thus, \( M(l_k) \neq M'(l_k) \).
\end{lemma}

\begin{lemma}
\label{lem:chap4-s-prefers-l-prefers1}
Let \( M \) and \( M' \) be stable matchings in an instance \( I \). 
If a student \( s_i \) is assigned to different projects in \( M \) and \( M' \), prefers \( M \) to \( M' \), and is assigned to \( p_j \) in \( M' \), then:
\begin{enumerate}[label=(\alph*)]
    \item If there exists a student in \( M(p_j) \setminus M'(p_j) \), then \( l_k \) prefers \( s_i \) to each student in \( M(p_j) \setminus M'(p_j) \).
    \item If \( p_j \) is undersubscribed in \( M \), then \( l_k \) prefers \( s_i \) to each student in \( M(l_k) \setminus M'(l_k) \).
\end{enumerate}
\end{lemma}

\section{Meta-rotations}
\label{sect:eliminatin-rotation}
In this section we formally define meta-rotations in {\sc spa-s} and show that successively identifying and eliminating exposed meta-rotations yields another stable matching of the instance. We start by defining the \emph{next project} of a student (Definition~\ref{def:next-project}), i.e., a project to which the student may be assigned in another stable matching of~$I$, and then define when a meta-rotation is said to be exposed (Definition~\ref{def:exposed-mr}).

\begin{definition}[Next project]
\label{def:next-project}
Let $M_L$ be the lecturer-optimal stable matching of an instance~$I$, and let $M$ be any stable matching with $M \ne M_L$. 
For a student~$s_i$ with $M(s_i) \ne M_L(s_i)$, let $p_j = M(s_i)$ and $l_k$ the lecturer offering~$p_j$. 
Denote by $w_M(p_j)$ the worst student assigned to~$p_j$ in~$M$, and by $w_M(l_k)$ the worst student assigned to~$l_k$ in~$M$. The \emph{next project} for~$s_i$, denoted $s_M(s_i)$, is the first project~$p$ on $s_i$’s preference list that appears after~$p_j$ and satisfies one of the following, where~$l$ is the lecturer offering~$p$:
\begin{enumerate}[label=(\roman*)]
    \item $p$ is full in~$M$ and $l$ prefers~$s_i$ to~$w_M(p)$; or
    \item $p$ is undersubscribed in~$M$, $l$ is full in~$M$, and $l$ prefers~$s_i$ to~$w_M(l)$.
\end{enumerate}
\end{definition}
Let $\mathit{next}_M(s_i)$ denote the \emph{next student} for~$s_i$. 
If $p$ satisfies (i), then $\mathit{next}_M(s_i) = w_M(p)$; if $p$ satisfies (ii), then $\mathit{next}_M(s_i) = w_M(l)$. 
We note that such a project may not exist. For instance, if \( M \) is the lecturer-optimal stable matching, no student can be assigned to a less preferred project in any other stable matching.

\medskip
To illustrate this, consider instance \(I_1\) in Figure \ref{fig:instance1}, which admits seven stable matchings, one of which is \(M_2 = \{(s_1, p_1), (s_2, p_1), (s_3, p_3), (s_4, p_3), (s_5, p_4), (s_6, p_5),\\  (s_7, p_7), (s_8, p_8), (s_9, p_2)\}\). It can be observed that the first project on \(s_6\)’s preference list following \(p_5\) (her assignment in \(M_2\)) is \(p_2\), which is full in \(M_2\). However, \(l_1\) (the lecturer offering \(p_2\)) prefers the worst student in \(M_2(p_2)\), namely \(s_9\), to \(s_6\). Proceeding to the next project, \(p_7\), which is full in \(M_2\), it is clear that \(l_2\) prefers \(s_6\) to the worst student in \(M_2(p_7)\), namely \(s_7\). Therefore, \(next_M(s_6) = s_7\). Similarly, \(p_6\) is the first project on \(s_7\)’s preference list that is undersubscribed in \(M_2\), and \(l_1\) prefers \(s_7\) to the worst student in \(M_2(l_1)\), namely \(s_6\). Thus, \(next_M(s_7) = s_6\).

\begin{figure}[h]
\centering
\renewcommand{\arraystretch}{1} 
\setlength{\tabcolsep}{7pt} 
\resizebox{\textwidth}{!}{ 
\begin{tabular}{p{0.33\textwidth} p{0.41\textwidth} p{0.25\textwidth}}
\hline
\textbf{Students' preferences} & \textbf{Lecturers' preferences} & \textbf{Offers} \\ 
\hline
$s_1$: $p_1 \; p_2 \; p_4 \; p_3$ & $l_1$: $s_7 \; s_9 \; s_3 \; s_4 \; s_5 \; s_1 \; s_2 \; s_6 \; s_8 $ & $p_1$, $p_2$, $p_5$, $p_6$ \\ 

$s_2$: $p_1 \; p_4 \; p_3 \; p_2$ & $l_2$: $s_6 \; s_1 \; s_2 \; s_5 \; s_3 \; s_4 \; s_7 \; s_8 \;  s_9$ & $p_3$, $p_4$, $p_7$, $p_8$ \\ 

$s_3$: $p_3 \; p_1 \; p_2 \; p_4$ & & \\ 
$s_4$: $p_3 \; p_2 \; p_1 \; p_4$ & & \\ 
$s_5$: $p_4 \; p_3 \; p_1$ & & \\ 
$s_6$: $p_5 \; p_2 \; p_7$ & & \\ 
$s_7$: $p_7 \; p_3 \; p_6$ & & \\ 
$s_8$: $p_6 \; p_8$ & \multicolumn{2}{l}{\textbf{Project capacities:} $c_1 = c_3 = 2$; $\forall j \in \{2, 4, 5, 6, 7, 8\}, \, c_j = 1$} \\ 
$s_9$: $p_8 \; p_2 \; p_3$ & \multicolumn{2}{l}{\textbf{Lecturer capacities:} $d_1 = 4$, $d_2 = 5$}  \\ 
\hline
\end{tabular}
}
\caption{An instance \( I_1 \) of {\sc spa-s}} 
\label{fig:instance1}
\end{figure}

\begin{table}[h]
    \centering
    \setlength{\tabcolsep}{4pt} 
    \renewcommand{\arraystretch}{1.1} 
    \begin{tabular}{c|*{9}{c}} 
        Matching & $s_1$ & $s_2$ & $s_3$ & $s_4$ & $s_5$ & $s_6$ & $s_7$ & $s_8$ & $s_9$ \\
        \hline
        $M_1$ & $p_1$ & $p_1$ & $p_3$ & $p_3$ & $p_4$ & $p_5$ & $p_7$ & $p_6$ & $p_8$ \\
        $M_2$ & $p_1$ & $p_1$ & $p_3$ & $p_3$ & $p_4$ & $p_5$ & $p_7$ & $p_8$ & $p_2$  \\
        $M_3$ & $p_1$ & $p_1$ & $p_3$ & $p_3$ & $p_4$ & $p_7$ & $p_6$ & $p_8$ & $p_2$ \\
        $M_4$ & $p_1$ & $p_4$ & $p_3$ & $p_1$ & $p_3$ & $p_5$ & $p_7$ & $p_8$ & $p_2$ \\
        $M_5$ & $p_1$ & $p_4$ & $p_3$ & $p_1$ & $p_3$ & $p_7$ & $p_6$ & $p_8$ & $p_2$\\
        $M_6$ & $p_4$ & $p_3$ & $p_1$ & $p_1$ & $p_3$ & $p_5$ & $p_7$ & $p_8$ & $p_2$  \\
        $M_7$ & $p_4$ & $p_3$ & $p_1$ & $p_1$ & $p_3$ & $p_7$ & $p_6$ & $p_8$ & $p_2$ \\
    \end{tabular}
    \caption{Instance $I_1$ admits seven stable matchings.}
    \label{tab:instance1}
\end{table}

\begin{definition}[{Exposed Meta-Rotation}]
\label{def:exposed-mr}
Let \( M \) be a stable matching, and let \(\rho = \{(s_{0}, p_{0}), (s_{1}, p_{1}), \ldots, (s_{r-1}, p_{r-1})\}\) be an ordered list of student–project pairs in \( M \), where \( r \ge 2 \). For each \( t \in \{0, \ldots, r-1\} \), let \( s_t \) be the worst student assigned to project \( p_t \) in \( M \), and let \( s_{t+1} = \mathit{next}_M(s_t) \) (indices taken modulo \( r \)).  
Then \( \rho \) is an \emph{exposed meta-rotation} in \( M \).
\end{definition}
Note that in any exposed meta-rotation \( \rho \) of a stable matching \( M \), each student and project that appears in \( \rho \) is part of an assigned pair in \( M \), and each appears exactly once in \( \rho \). This is because, in \( M \), each project has a unique worst student among those assigned to it, and the definition of \( \rho \) includes precisely one such student--project pair. Furthermore, the set of all meta-rotations in \( I \) consists precisely of those ordered sets of pairs that are exposed in at least one stable matching \( M \in \mathcal{M} \).

\begin{definition}[\textbf{Meta-rotation Elimination}]
\label{def:elimination-of-rho}
Given a stable matching \( M \) and an exposed meta-rotation \( \rho \) in \( M \), we denote by \( M/\rho \) the matching obtained by assigning each student \( s \in \rho \) to project \( s_M(s) \), while keeping the assignments of all other students unchanged. This transition from \( M \) to \( M/\rho \) is referred to as the \emph{elimination} of \( \rho \) from \( M \).
\end{definition}

\subsection{Justification for the meta-rotation definition}
In both {\sc sm} and {\sc hr}, an exposed rotation \( \rho \) in a stable matching \( M \) is defined as a sequence of pairs such that performing a cyclic shift yields a new stable matching \( M/\rho \). In {\sc sm}, each woman is assigned to the next man in the sequence, and in {\sc hr}, each hospital is assigned to the next resident. Specifically, in {\sc hr}, if some resident \(r\), who is assigned in a stable matching \(M\), has a \emph{next} hospital \(h\) on their preference list and is part of an exposed rotation \(\rho\), then \(r\) swaps places with the least preferred resident currently assigned to \(h\) in \(M\), forming the new matching \(M/\rho\). Moreover, by the Rural Hospitals Theorem for {\sc hr},  if some hospital \(h\) is undersubscribed in one stable matching, then it is assigned the same set of residents across all stable matchings.

\medskip 
However, as we noted earlier, these properties do not extend to {\sc spa-s} for projects or lecturers that are undersubscribed. In {\sc spa-s}, the number of students assigned to a project may vary across stable matchings. Consequently, a project that is part of an exposed meta-rotation \( \rho \) in a given stable matching \( M \) may not necessarily appear in the resulting stable matching \( M/\rho \). For example, in instance \( I_3 \) from Figure~\ref{fig:instance1}, the pairs \( \{(s_6, p_5), (s_7, p_7)\} \) form an exposed meta-rotation in \( M_2 \). Here, project \( p_5 \) is full in \( M_2 \) but becomes undersubscribed in \( M_3 \). Clearly, neither \( p_5 \) nor its lecturer \( l_1 \) (who offers \( p_5 \)) have the same set of assigned students in \( M_2 \) and \( M_3 \). Nevertheless, by the \emph{Unpopular Projects Theorem} (see Theorem \ref{thm:unpopular-students}), the total number of students assigned to each lecturer remains the same across all stable matchings.

\medskip
To address these differences, our definition of meta-rotations explicitly accounts for whether each project is full or undersubscribed in the stable matching of interest. Suppose a student \( s_i \), assigned to some project in a stable matching \( M \), has $p_j$ as their next possible project. Whether \( s_i \) can be assigned to \( p_j \) in another stable matching depends on the status of \( p_j \) in \( M \) as well as the preference of the lecturer \( l_k \) who offers it. If \( p_j \) is full in \( M \), then the assignment of \( s_i \) to \( p_j \) is possible only if \( l_k \) prefers \( s_i \) to the worst student in \( M(p_j) \); in this case, \( s_i \) takes the place of that student. If \( p_j \) is undersubscribed in \( M \), then the assignment is possible only if \( l_k \) prefers \( s_i \) to the worst student in $M(l_k)$; here, \( s_i \) is assigned to \( p_j \) and the least preferred student in \( M(l_k) \) is removed. These conditions ensure that each such assignment yields a new matching that is stable.

\subsection{Identifying meta-rotations}
Here, we establish results concerning meta-rotations. In Lemma~\ref{lem:no-p-between-M(s_i)-and-sM(s_i)}, we show that for any pair \((s_t,p_t)\) in a meta-rotation \(\rho = \{(s_0,p_0), \ldots, (s_{r-1},p_{r-1})\}\), no project that lies strictly between \(p_t\) and the next project \(s_M(s_t)\) in \(s_t\)'s preference list can form a stable pair \footnote{A stable pair is one that occurs in some stable matching of the instance}. 
In Lemma~\ref{lem:one-metarotation-in-M}, we show that every stable matching other than the lecturer-optimal stable matching \(M_L\) contains at least one exposed meta-rotation. In Lemma~\ref{lem:s_r-and-s_z}, we show that when constructing \(M/\rho\), if a student becomes assigned to a lecturer \(l_k\), then \(l_k\) simultaneously loses a student from \(M(l_k)\). Finally, in Lemma~\ref{lem:M-rho-is-stable}, we prove that if a meta-rotation \(\rho\) is exposed in a stable matching \(M\), then the matching \(M/\rho\), obtained by eliminating \(\rho\), is also stable, and that \(M\) dominates \(M/\rho\).

\begin{lemma}
\label{lem:no-p-between-M(s_i)-and-sM(s_i)}
Let \( \rho = \{(s_{0}, p_{0}), (s_{1}, p_{1}), \ldots, (s_{r-1}, p_{r-1})\} \) be an exposed meta-rotation in a stable matching \( M \) for instance \( I \). Suppose that for some student \( s_t \) (where \( 0 \leq t \leq r-1 \)), there exists a project \( p_z \) such that \( s_t \) prefers \( p_t \) to \( p_z \), and prefers \( p_z \) to \( s_M(s_t) \). Then \( (s_t, p_z) \) is not a stable pair.
\end{lemma}

\begin{proof}
Let \( M \) be a stable matching in which the meta-rotation \( \rho \) is exposed, and suppose that \( (s_i, p_j) \in \rho \). Suppose there exists a project \( p_z \) on \( s_i \)'s preference list such that \( s_i \) prefers \( p_j \) to \( p_z \), and prefers \( p_z \) to \( s_M(s_i) \). Let \( l_z \) be the lecturer who offers \( p_z \), and possibly also offers \( s_M(s_i) \). Suppose for contradiction that there exists another stable matching \( M' \) in which \( s_i \) is assigned to \( p_z \), that is, \( s_i \in M'(p_z) \setminus M(p_z) \). Then \( s_i \) prefers \( M \) to \( M' \). Since \( p_z \neq s_M(s_i) \), by definition of \( s_M(s_i) \), one of the following conditions holds in \( M \):
\begin{enumerate}[label = (\roman*)]
    \item both \( p_z \) and \( l_z \) are undersubscribed,
    \item \( p_z \) is full and \( l_z \) prefers the worst student in \( M(p_z) \) to \( s_i \), or
    \item \( p_z \) is undersubscribed, \( l_z \) is full, and \( l_z \) prefers the worst student in \( M(l_z) \) to \( s_i \).
\end{enumerate}

\noindent \textit{Case (i):}  Both \( p_z \) and \( l_z \) are undersubscribed in \( M \). Then \( l_z \) is undersubscribed in \( M' \) since \( |M(l_z)| = |M'(l_z)| \). Moreover, by Theorem \ref{thm:unpopular-students}, since \( p_z \) is offered by an undersubscribed lecturer $l_z$, then \( |M(p_z)| = |M'(p_z)| \), meaning $p_z$ is undersubscribed in $M'$. Since \( s_i \in M'(p_z) \setminus M(p_z) \), there exists a student \( s_z \in M(p_z) \setminus M'(p_z) \). If \( s_z \) prefers \( M \) to \( M' \), then \( (s_z, p_z) \) blocks \( M' \), as \( p_z \) and \( l_z \) are undersubscribed in \( M' \). Therefore, \( s_z \) prefers \( M' \) to \( M \). By the first part of Lemma~\ref{lem:chap4-s-prefers-l-prefers1}, since $s_z$ prefers $M'$ to $M$ and \( s_i \in M'(p_z) \setminus M(p_z) \), then \( l_z \) prefers \( s_z \) to \( s_i \). However, by the same lemma, since \( s_i \) prefers \( M \) to \( M' \) and \( s_z \in M(p_z) \setminus M'(p_z) \), then \( l_z \) prefers \( s_i \) to \( s_z \). This gives a direct contradiction, as \( l_z \) cannot simultaneously prefer \( s_i \) to \( s_z \) and \( s_z \) to \( s_i \). Hence, case (i) cannot occur.

\vspace{0.2cm}
\textit{Case (ii):} Suppose \( p_z \) is full in \( M \) and \( l_z \) prefers the worst student in \( M(p_z) \) to \( s_i \). Since \( s_i \in M'(p_z) \setminus M(p_z) \) and $p_z$ is full in $M$, there exists some student \( s_z \in M(p_z) \setminus M'(p_z) \). Thus, \( l_z \) prefers \( s_z \) to \( s_i \). However, by Lemma~\ref{lem:chap4-s-prefers-l-prefers1}, since \( s_i \) prefers \( M \) to \( M' \) and \( s_z \in M(p_z) \setminus M'(p_z) \), \( l_z \) prefers \( s_i \) to \( s_z \). This yields a direct contradiction on $l_k$'s preferences similar to case~(i). Hence, case~(ii) cannot occur.

\vspace{0.2cm}
\textit{Case (iii):} Suppose \( p_z \) is undersubscribed in \( M \), \( l_z \) is full in \( M \), and \( l_z \) prefers the worst student in \( M(l_z) \) to \( s_i \). This implies that \( l_z \) prefers each student in \( M(l_z) \) to \( s_i \). We claim that there exists some student \( s_z \in M(l_z) \setminus M'(l_z) \). If \( s_i \) is assigned to different projects offered by \( l_z \) in both \( M \) and \( M' \), then by Lemma~\ref{lem:chap4samelecturer1}, there exists some student \( s \in M'(l_z) \setminus M(l_z) \). Consequently, we have some \( s_z \in M(l_z) \setminus M'(l_z) \), since $|M(l_k)| = |M'(l_k)|$. The same conclusion holds if \( s_i \in M'(l_z) \setminus M(l_z) \). Thus, it follows that \( l_z \) prefers \( s_z \) to \( s_i \). However, by Lemma~\ref{lem:chap4-s-prefers-l-prefers1}, since \( s_i \) prefers \( M \) to \( M' \) and \( p_z \) is undersubscribed in \( M \), we have that \( l_z \) prefers \( s_i \) to \( s_z \). This yields a direct contradiction in \( l_z \)'s preference, as in case~(i).

\vspace{0.1cm}
Since all possible cases lead to a contradiction, the pair \( (s_i, p_z) \) does not belong to any stable matching of \( I \), completing the proof.
\end{proof}

\noindent The following corollary follows immediately from Lemma~\ref{lem:no-p-between-M(s_i)-and-sM(s_i)}:
\begin{corollary}
\label{cor:p*-l*-undersubscribed}
Let \( M \) be a stable matching in \( I \), and let \( s_i \) be a student for whom \( s_M(s_i) \) exists. Suppose that \( s_i \) prefers \( M(s_i) \) to some project \( p_z \) offered by lecturer \( l_z \), and prefers \( p_z \) to \( s_M(s_i) \). If both \( p_z \) and \( l_z \) are undersubscribed in \( M \), then the pair \( (s_i, p_z) \) does not appear in any stable matching of \( I \).
\end{corollary}

\smallskip
\begin{lemma}
\label{lem:one-metarotation-in-M}
Let \( M \) be a stable matching in an instance of {\sc spa-s}, and suppose \( M \neq M_L \), where \( M_L \) is the lecturer-optimal stable matching. Then there exists at least one meta-rotation that is exposed in \( M \). 
\end{lemma}

\begin{proof}
Let \( M \) be a stable matching in an instance $I$ of {\sc spa-s}, and let $M_L$ be the lecturer-optimal stable matching. Clearly, $M$ dominates $M_L$. Since $M \neq M_L$, there exists some student $s_{i_0}$, who is assigned to different projects in $M$ and $M_L$. Suppose that $s_{i_0}$ is assigned to $p_{j_0}$ in $M$ and assigned to $p_{t_0}$ in $M_L$, where \( l_t \) offers \( p_{t_0} \) (possibly $l_t$ offers both $p_{j_0}$ and $p_{t_0}$). Clearly, \( s_{i_0} \) prefers $p_{j_0}$ to $p_{t_0}$. Furthermore, \( p_{t_0} \) is either (i) undersubscribed in \( M \) or (ii) full in \( M \). In both cases, we will prove that \( s_M(s_{i_0}) \) exists, which in turn proves the existence of \( next_M(s_{i_0}) \).

\medskip
First, suppose that \( p_{t_0} \) is undersubscribed in \( M \). By Lemma \ref{lem:s-prefers-M-l-prefers-M':dominance}, \( l_t \) prefers \( s_{i_0} \) to the worst student in \( M(l_t) \). Furthermore, by Lemma \ref{lem:p*-l*-undersubscribed}, if \( p_{t_0} \) is undersubscribed in \( M \), then \( l_t \) must be full in \( M \). Given that $s_{i_0}$ prefers $p_{j_0}$ to $p_{t_0}$, \( p_{t_0} \) is undersubscribed in \( M \), \( l_t \) is full in \( M \), and \( l_t \) prefers \( s_{i_0} \) to the worst student in \( M(l_t) \), it follows that \( s_M(s_{i_0}) \) exists.  Now, consider case (ii), where \( p_{t_0} \) is full in \( M \). Since \( s_{i_0} \) is assigned to \( p_{t_0} \) in \( M_L \) and \( p_{t_0} \) is full in \( M \), by Lemma \ref{lem:s-prefers-M-l-prefers-M':dominance}, we have that \( l_t \) prefers \( s_{i_0} \) to the worst student in \( M(p_{t_0}) \). Since these condition hold, \( s_M(s_{i_0}) \) exists, and consequently, \( \mathit{next}_M(s_{i_0}) \) exists.

\medskip
Let \( next_M(s_{i_0}) = s_{i_1} \). By definition, \( s_{i_1} \) is either the worst student assigned to \( p_{t_0} \) in \( M \) (if \( p_{t_0} \) is full in \( M \)), or the worst student assigned to \( l_t \) in \( M \) (if \( p_{t_0} \) is undersubscribed in \( M \)). In either case, \( l_t \) prefers \( s_{i_0} \) to \( s_{i_1} \). Furthermore, since \( s_{i_0} \) is assigned to \( p_{j_0} \) in \( M \) and to \( p_{t_0} \) in \( M_L \), it follows from Lemma~\ref{lem:swaps} that the worst student in \( M(p_{t_0}) \) is not in \( M_L(p_{t_0}) \) (if $p_{t_0}$ is full in $M$), and the worst student in \( M(l_t) \) is not in \( M_L(l_t) \) (if $p_{t_0}$ is undersubscribed in $M$). Therefore, \( s_{i_1} \) is assigned to different projects in \( M \) and \( M_L \). Let \( p_{j_1} = M(s_{i_1}) \), where \( l_t \) offers \( p_{j_1} \) (possibly \( p_{t_0} = p_{j_1} \)).  Let \( p_{t_1} = M_L(s_{i_1}) \), and let \( l_{t_1} \) be the lecturer who offers \( p_{t_1} \) (possibly \( l_t = l_{t_1} \)). Clearly, $s_{i_1}$ prefers $p_{j_1}$ to $p_{t_1}$. Again, it follows that \( p_{t_1} \) is either (i) undersubscribed in $M$ or (ii) full in \( M \).  Following a similar argument as before, we will prove that both \( s_M(s_{i_1}) \) and \( next_M(s_{i_1}) \) exist. 

\medskip
First, suppose that \( p_{t_1} \) is undersubscribed in \( M \). By Lemma \ref{lem:s-prefers-M-l-prefers-M':dominance}, \( l_{t_1} \) prefers \( s_{i_1} \) to the worst student in \( M(l_{t_1}) \). Furthermore, by Lemma \ref{lem:p*-l*-undersubscribed}, if $p_{t_1}$ is undersubscribed in $M$, then \( l_{t_1} \) must be full in \( M \). Given that $s_{i_1}$ prefers $p_{j_1}$ to $p_{t_1}$, $p_{t_1}$ is undersubscribed in $M$, $l_{t_1}$ is full in $M$, and $l_{t_1}$ prefers $s_{i_1}$ to the worst student in $M(l_{t_1})$, it follows that $s_M(s_{i_1})$ exists. Now, consider case (ii), where $p_{t_1}$ is full in $M$. Since $s_{i_1}$ is assigned to $p_{t_1}$ in $M_L$ and $p_{t_1}$ is full in $M$, by Lemma \ref{lem:s-prefers-M-l-prefers-M':dominance}, we have that \( l_{t_1} \) prefers \( s_{i_1} \) to the worst student in \( M(p_{t_1}) \). Since this condition holds, \( s_M(s_{i_1}) \) exists, and consequently, \( \mathit{next}_M(s_{i_1}) \) exists. 

\medskip
Let \( next_M(s_{i_1}) = s_{i_2} \). By definition, \( s_{i_2} \) is either the worst student assigned in \( M(p_{t_1}) \) if \( p_{t_1} \) is full in \( M \), or the worst student in \( M(l_{t_1}) \) if \( p_{t_1} \) is undersubscribed in \( M \).  In either case, \( l_{t_1} \) prefers \( s_{i_1} \) to \( s_{i_2} \). Furthermore, since \( s_{i_1} \) is assigned to \( p_{j_1} \) in \( M \) and to \( p_{t_1} \) in \( M_L \), it follows from Lemma~\ref{lem:swaps} that the worst student in \( M(p_{t_1}) \) is not in \( M_L(p_{t_1}) \) (if $p_{t_1}$ is full in $M$), and the worst student in \( M(l_{t_1}) \) is not in \( M_L(l_{t_1}) \) (if $p_{t_1}$ is undersubscribed in $M$). Therefore, \( s_{i_2} \) is assigned to different projects in \( M \) and \( M_L \). 
Let \( p_{j_2} = M(s_{i_2}) \), where \( l_{t_1} \) offers \( p_{j_2} \) (possibly \( p_{j_2} = p_{t_1} \)). Let \( p_{t_2} = M_L(s_{i_2}) \), and let \( l_{t_2} \) be the lecturer who offers \( p_{t_2} \). Clearly, $s_{i_2}$ prefers $p_{j_2}$ to $p_{t_2}$.  Again, it follows that \( p_{t_2} \) is either (i) undersubscribed in $M$ or (ii) full in \( M \). Following a similar argument as in the previous paragraphs, both $s_M(s_{i_2})$ and $next_M(s_{i_2})$ exist.

\medskip
By continuing this process, we observe that each identified student-project pair \((s_i, p_j)\) in \( M \) leads to another pair in \( M \), which in turn leads to another pair, and so forth, thereby forming a sequence of pairs \((s_{i_0}, p_{j_0}), (s_{i_1}, p_{j_1}), \dots \) within \( M \) such that $s_{i_1}$ is $next_M(s_{i_0})$, $s_{i_2}$ is $next_M(s_{i_1})$, and so on. Moreover, each student that we identify is assigned to different projects in $M$ and $M_L$, and prefers their assigned project in $M$ to $M_L$. Given that the number of students in $M$ is finite, this sequence cannot extend indefinitely and must eventually terminate with a pair in $M$ that we have previously identified. 

\medskip
Suppose that \((s_{i_{r-1}}, p_{j_{r-1}})\) is the final student-project pair identified in this sequence, let \( s_{i_r} \) be \(\text{next}_M(s_{i_{r-1}}) \), and let $M(s_{i_r})$ be $p_{j_r}$. It follows that \( s_{i_r} \) must have appeared earlier in the sequence. Otherwise, we would need to extend the sequence by including the pair, $(s_{i_r},p_{j_r})$, contradicting the assumption that \((s_{i_{r-1}}, p_{j_{r-1}})\) is the last pair identified in the sequence. Therefore, at some point, a student-project pair must reappear in the sequence, and when this occurs, the process terminates. As an example, suppose that the sequence starts with \( (s_{i_0}, p_{j_0}) \), and that the last pair \( (s_{i_r}, p_{j_r}) \) satisfies \( s_{i_r} = s_{i_1} \). Then, the subsequence \( \{(s_{i_1}, p_{j_1}), (s_{i_2}, p_{j_2}), \dots, (s_{i_{r-1}}, p_{j_{r-1}})\} \) forms an exposed meta-rotation in \( M \), as illustrated in Figure~\ref{fig:exposed-metarotation in M}.
\end{proof}

\begin{figure}[htb]
\centering
\resizebox{0.9\textwidth}{!}{%
  \begin{tikzpicture}[->, >=stealth, node distance=1.5cm, thick]
    \node (A) {$(s_{i_0},p_{j_0})$};
    \node (B) [right=of A] {$(s_{i_1},p_{j_1})$};
    \node (C) [right=of B] {$(s_{i_2},p_{j_2})$};
    \node (D) [right=of C] {$\cdots$};
    \node (E) [right=of D] {$(s_{i_{r-1}},p_{j_{r-1}})$};
    
    \draw[->] (A) -- (B);
    \draw[->] (B) -- (C);
    \draw[->] (C) -- (D);
    \draw[->] (D) -- (E);
    \draw[->] (E) edge[bend left=25] (B);
  \end{tikzpicture}
}
\caption{An exposed meta-rotation in $M$.}
\label{fig:exposed-metarotation in M}
\end{figure}

The proof of Lemma~\ref{lem:one-metarotation-in-M} gives a constructive method for identifying an exposed meta-rotation in any stable matching \( M \) of a {\sc spa-s} instance~\( I \). Define a directed graph \( H(M) \) whose vertices are the students assigned to different projects in \( M \) and \( M_L \). 
For each such student \( s_i \), add a directed edge from \( s_i \) to \( \mathit{next}_M(s_i) \); by construction, every vertex has exactly one outgoing edge. 
Since the number of vertices is finite, \( H(M) \) must contain at least one directed simple cycle, which corresponds to the students involved in an exposed meta-rotation in \( M \). 
To identify it, start from any vertex and follow its outgoing edges until a vertex repeats; the students encountered from the first to the second occurrence of that vertex form the exposed meta-rotation.

\begin{corollary}
\label{cor:each-stud-in-one-rotation}
Let \( M \) be a stable matching different from the lecturer-optimal matching \( M_L \), and let \( H(M) \) be the directed graph whose vertices are the students assigned to different projects in \( M \) and \( M_L \). Then:
\begin{enumerate}[label=(\roman*)]
    \item each vertex \( s_i \in H(M) \) has exactly one outgoing edge;
    \item starting from any vertex \( s_i \in H(M) \), there is a unique directed path in \( H(M) \)  that terminates at the last student of some exposed meta-rotation \( \rho \) in \( M \); and
    \item every student in \( H(M) \) either belongs to exactly one exposed meta-rotation in \( M \) or lies on the path leading to one.
\end{enumerate}
\end{corollary}
 
\noindent \textbf{Example:} Consider instance \( I \) in Figure \ref{figure:leadstorotation}, where the student-optimal stable matching is
\( M = \{(s_1, p_1), (s_2, p_3), \) \( (s_3, p_2), (s_4, p_4)\} \), and the lecturer-optimal stable matching is  \( M_L = \{(s_1, p_2), (s_2, p_4), \)
\( (s_3, p_1), (s_4, p_3)\} \).
Each student is assigned to different projects in $M$ and $M_L$, and for each student, we have: $next_M(s_1) = s_3, ~next_M(s_2) = s_4, ~next_M(s_3) = s_1, ~next_M(s_4) = s_1$.  The directed graph \( H(M) \) corresponding to \( M \) is shown in Figure \ref{fig:H(M)}. Starting at $s_2$, the sequence of visited students is: $s_2 \rightarrow s_4 \rightarrow s_1 \rightarrow s_3 \rightarrow s_1$. Since $s_1$ appears twice, the first cycle in this sequence is determined by the students from the first occurrence of $s_1$
up to (but not including) its second occurrence. Thus, the students forming the meta-rotation are $s_1$ and $s_3$, and the corresponding meta-rotation exposed in $M$ is $\rho = \{(s_1, p_1), (s_3, p_2)\}$.

\begin{figure}[!h]
    \centering
    \begin{tikzpicture}[->, >=stealth, node distance=1.5cm, thick]
        \node (s2) [circle, draw, minimum size=0.5cm] {$s_2$};
        \node (s4) [circle, draw, minimum size=0.5cm, right of=s2] {$s_4$};
        \node (s1) [circle, draw, minimum size=0.5cm, below of=s4] {$s_1$};
        \node (s3) [circle, draw, minimum size=0.5cm, left of=s1] {$s_3$};

        \draw[->] (s2) -- (s4);
        \draw[->] (s4) -- (s1);
        \draw[->] (s1) -- (s3);
        \draw[->] (s3) to[out=25, in=160, looseness=1] (s1); 
    \end{tikzpicture}
    \caption{Graph $H(M)$ for $M$}
    \label{fig:H(M)}
\end{figure}

\vspace{0.2cm}
We observe that a student \( s_i \) may be assigned different projects in \( M \) and \( M_L \) without being part of an exposed meta-rotation \( \rho \) in \( M \). In such a case, if there exists a directed path from \( s_i \) to some student involved in \( \rho \), we say that \( s_i \) \textit{leads to} \( \rho \). For instance, \( s_4 \in M_L(l_4) \setminus M(l_4) \) and \( s_4 \notin \rho \), so \( s_4 \) \textit{leads to} \( \rho \).

\begin{lemma}
\label{lem:s_r-and-s_z}
Let $M$ be a stable matching in $I$ different from the lecturer-optimal matching $M_L$ and let $\rho$ be an exposed meta-rotation in $M$. If some student $s_i \in \rho$ such that $s_M(s_i)$ is offered by lecturer $l_k$, then there exists some other student $s_z \in M(l_k)$ such that  $l_k$ prefers $s_i$ to $s_z$, $s_z \in \rho$, and $s_M(s_z)$ is offered by a lecturer different from $l_k$.
\end{lemma}
\begin{proof}
Let \( M \) be a stable matching with an exposed meta-rotation \(\rho\). Suppose there exists some student \( s_{i_0} \in \rho \), such that \( s_M(s_{i_0}) \) is offered by lecturer \( l_k \). Without loss of generality, suppose that $(s_{i_0},p_{j_0})$ is the first pair in $\rho$. Now suppose for a contradiction that there exists no student $s_z \in M(l_k)$, such that $s_z \in \rho$ and $s_M(s_z)$ is offered by a lecturer different from $l_k$. Since \( s_{i_0} \in \rho \), there exists a student \( s_{i_1} \in \rho \) where \( s_{i_1} = \text{next}_M(s_{i_0}) \) and, by definition of $next_M(s_{i_0})$, \( l_k \) prefers \( s_{i_0} \) to \( s_{i_1} \). Hence, $s_M(s_{i_1})$ exists and by our assumption, $s_M(s_{i_1})$ is offered by $l_k$. Similarly, since \( s_{i_1} \in \rho \), there exists a student \( s_{i_2} \in \rho \) with \( s_{i_2} = \text{next}_M(s_{i_1}) \) and \( l_k \) prefers \( s_{i_1} \) to \( s_{i_2} \). Again, $s_M(s_{i_2})$ is also offered by $l_k$. Continuing in this manner, we obtain a sequence of student-project pairs \(
(s_{i_0},p_{j_0}), (s_{i_1},p_{j_1}), (s_{i_2},p_{j_2}), \dots, (s_{i_{r-1}},p_{j_{r-1}}), (s_{i_r},p_{j_r})\) in \(\rho\) such that for each \( t \) with \( 0 \le t < r \):

\begin{itemize}
    \item \( s_{i_{t+1}} = \text{next}_M(s_{i_t}) \),
    \item \( l_k \) prefers \( s_{i_t} \) to \( s_{i_{t+1}} \), and
    \item \( s_M(s_{i_{t+1}}) \) is offered by \( l_k \).
\end{itemize}

\noindent Since \(\rho\) is finite, this sequence cannot continue indefinitely and we would identify some student-project pair that appeared earlier in the sequence. Without loss of generality, let (\( s_{i_r},p_{j_r} \)) be the first pair to reappear in the sequence. By construction, $s_{i_r}$ is $next_M(s_{i_{r-1}})$, $l_k$ prefers $s_{i_{r-1}}$ to $s_{i_r}$, and $s_M(s_{i_r})$ is offered by $l_k$. Clearly, $s_{i_r} \neq s_{i_{r-1}}$. Therefore, $s_{i_r}$ must have appeared earlier in the sequence before $s_{i_{r-1}}$.  However, since $s_{i_r}$ appears earlier in the sequence, then $s_{i_r}$ must be some student that $l_k$ prefers to $s_{i_{r-1}}$. This yields a contradiction since we assume that $l_k$ prefers $s_{i_{r-1}}$ to $s_{i_r}$. Therefore, there exists at least one student $s_z \in M(l_k)$, where $s_z \in \rho$ and
$s_M(s_z)$ is offered by a lecturer different from $l_k$.
\end{proof}

\begin{lemma}
\label{lem:M-rho-is-stable}
If \(\rho\) is a meta-rotation exposed in a stable matching \(M\), then the matching obtained by eliminating \(\rho\) from \(M\), denoted as \(M/\rho\), is a stable matching. Furthermore, \(M\) dominates \(M/\rho\).
\end{lemma}

\vspace{-0.3cm}
\begin{proof}
Let \( M \) be a stable matching in which \(\rho\) is exposed, and let \( M' = M/\rho \) denote the matching obtained by eliminating \(\rho\). 
By definition, only students in \(\rho\) change projects; each \( s_i \in \rho \) moves from \( M(s_i) \) to \( s_M(s_i) \), while all other students retain their projects in $M$. 
Hence, every student in \( M' \) is assigned to exactly one project. Consider any project \( p_j \) for which \( M'(p_j) \neq M(p_j) \). If \( p_j \) is full in \( M \), it loses exactly one student—the worst in \( M(p_j) \)—and gains one new student, so \( |M'(p_j)| = |M(p_j)| \). If \( p_j \) is undersubscribed in \( M \), then $l_k$ loses the worst student in \( M(l_k) \), and \( p_j \) gains one student, so \( |M'(p_j)| \ge |M(p_j)| \). Hence, no project is oversubscribed in \( M' \).

\medskip
We now show that no lecturer is oversubscribed in \( M' \). 
Since \(\rho\) is exposed in \( M \), for each student \( s_i \in \rho \) assigned to a project offered by lecturer \( l \), by Lemma~\ref{lem:s_r-and-s_z}, there exists another student \( s_z \in \rho \) such that \( s_z \in M(l) \), \( l \) prefers \( s_i \) to \( s_z \), and \( s_M(s_z) \) is offered by a different lecturer. Thus, when \( s_i \) becomes assigned to \( l \) in \( M' \), \( s_z \) is simultaneously removed from \( M'(l) \). Hence, for every lecturer \( l_k \), \( |M'(l_k)| = |M(l_k)| \), and no lecturer is oversubscribed. Since each student is assigned to exactly one project and no capacity is exceeded, \( M' \) is a valid matching.

\medskip
Now, suppose that \( M' \) is not stable. Then there exists a blocking pair \((s_i, p_j)\) in \( M' \).  By the construction of \( M' \), if \( s_i \) is assigned in \( M' \), then \( s_i \) must also be assigned in \( M \). Let $M(s_i)$ be $p_a$ and let $M'(s_i)$ be $p_b$. Then, there are three possible conditions on student \( s_i \):

\begin{itemize}
    \item [(S1):] \( s_i \) is unassigned in both \( M \) and \( M' \);
    \item [(S2):] \( s_i \) is assigned in both \( M \) and \( M' \), and prefers \( p_j \) to both \( p_a \) and \( p_b \);
    \item [(S3):] \( s_i \) is assigned in both \( M \) and \( M' \), \( s_i \) prefers \( p_a \) to \( p_j \), and prefers \( p_j \) to \( p_b \). 
\end{itemize}

\noindent Also, there are four possible conditions on $p_j$ and $l_k$:
\begin{enumerate}[label=(\roman*)]
    \item [(P1):] both \( p_j \) and \( l_k \) are undersubscribed in \( M' \);
    \item [(P2):] \( p_j \) is full in \( M' \) and \( l_k \) prefers \( s_i \) to the worst student in \( M'(p_j) \);
    \item [(P3):] \( p_j \) is undersubscribed in \( M' \), \( l_k \) is full in $M'$, and \( s_i \in M'(l_k) \);
    \item [(P4):] \( p_j \) is undersubscribed in \( M' \), \( l_k \) is full in $M'$, and \( l_k \) prefers \( s_i \) to the worst student in \( M'(l_k) \).
\end{enumerate}

\noindent \textbf{Case (S1 \& P1) and (S2 \& P1):} 
We claim that both $p_j$ and $l_k$ are undersubscribed in $M$. By the construction of \( M' \), every lecturer is assigned at least as many students in \( M' \) as in \( M \), that is, $|M(l_k)| = |M'(l_k)|$; thus, if \( l_k \) is undersubscribed in \( M' \), then \( l_k \) is undersubscribed in \( M \) as well. Similarly, if \( p_j \) is undersubscribed in \( M' \), then \( p_j \) is undersubscribed in \( M \), since by construction, $|M(p_j)| \leq |M'(p_j)|$. If \( s_i \) is unassigned in \( M \) or prefers \( p_j \) to \( M(s_i) \), the pair \((s_i, p_j)\) blocks \( M \), contradicting the stability of \( M \).

\vspace{0.1cm}
\noindent \textbf{Case (S3 \& P1):}
Following a similar argument as in (S1 \& P1) and (S2 \& P1), it follows that both $p_j$ and $l_k$ are undersubscribed in $M$. Since $s_i \in \rho$, \( s_i \) prefers \( p_a \) to \( p_j \), and prefers $p_j$ to $p_b$, then by Lemma~\(\ref{lem:no-p-between-M(s_i)-and-sM(s_i)}\), \((s_i, p_j)\) is not a stable pair. Hence, this case is impossible.

\vspace{0.1cm}
\noindent \textbf{Cases (S1 \& P2) and (S2 \& P2):} 
We claim that if \( p_j \) is full in \( M \), then \( l_k \) prefers \( s_i \) to the worst student in \( M(p_j) \); otherwise, if \( p_j \) is undersubscribed, \( l_k \) prefers \( s_i \) to the worst student in \( M(l_k) \). To see this, note that by the construction of \( M' \), one of three situations must occur. 
Either (i) \( p_j \) has the same set of students in \( M \) and \( M' \), in which case it is full and $l_k$ prefers $s_i$ to the worst student in $M(p_j)$; (ii) $M(p_j) \ne M'(p_j)$, and \( l_k \) prefers some student in \( M'(p_j) \) to the worst student in \( M(p_j) \), which means that $l_k$ prefers $s_i$ to the worst student in $M'(p_j)$; or (iii) \( p_j \) is undersubscribed in \( M \), and \( l_k \) prefers some student in \( M'(l_k) \) to the worst student in \( M(l_k) \), again implying that $l_k$ prefers $s_i$ to the worst student in $M(l_k)$. Hence our claim holds. Since \( s_i \) is either unassigned in both \( M \) and \( M' \) or prefers \( p_j \) to both $p_a$ and $p_b$, \( (s_i, p_j) \) blocks \( M \), a contradiction.

\vspace{0.1cm}
\noindent \textbf{Case (S3 \& P2):}
In this case, \( s_i \) prefers \( p_a \) to \( p_j \) and prefers \( p_j \) to \( p_b \). By applying a similar argument as in Cases (S1 \& P2) and (S2 \& P2), we conclude that either \( l_k \) prefers \( s_i \) to the worst student in \( M(p_j) \) if \( p_j \) is full in \( M \), or \( l_k \) prefers \( s_i \) to the worst student in \( M(l_k) \) if \( p_j \) is undersubscribed in \( M \). First, if \( p_j \) is full in \( M \), and \( l_k \) prefers \( s_i \) to the worst student in \( M(p_j) \), it follows directly from the definition of \( s_M(s_i) \) that \( p_j \) should be a valid \( next_M(s_i) \). Consequently, we should have \( M'(s_i) = p_j \), yielding a contradiction. Similarly, if \( p_j \) is undersubscribed in \( M \) and \( l_k \) prefers \( s_i \) to the worst student in \( M(l_k) \), then by the definition of \( s_M(s_i) \), \( p_j \) must be a valid \( next_M(s_i) \), which implies \( M'(s_i) = p_j \), another contradiction. Therefore, this blocking pair cannot occur in $M'$.

\vspace{0.1cm}
\noindent \textbf{Cases (S1 \& P3) and (S2 \& P3):} 
We claim that \( p_j \) is undersubscribed in \( M \), \( l_k \) is full in \( M \), and either \( s_i \in M(l_k) \) or \( l_k \) prefers \( s_i \) to the worst student in \( M(l_k) \). 
By the construction of \( M' \), one of two situations must occur. 
Either (i) \( M(l_k) = M'(l_k) \), in which case \( p_j \) is undersubscribed in \( M \), \( l_k \) is full in \( M \), and \( s_i \in M(l_k) \); 
or (ii) \( M(l_k) \neq M'(l_k) \), where some student in \( M'(l_k) \) is preferred by \( l_k \) to the worst student in \( M(l_k) \). 
Since \( |M(p_j)| \le |M'(p_j)| \) and $p_j$ is undersubscribed in $M'$, it follows that \( p_j \) is undersubscribed in $M$. Moreover, by construction of $M'$, $|M(l_k)| = |M'(l_k)|$, so $l_k$ is full in $M$.
Moreover, since \( l_k \) prefers \( s_i \) to the worst student in \( M'(l_k) \) (and prefers some student in $M'(l_k)$ to the worst student in $M(l_k)$ ), they prefer \( s_i \) to the worst student in \( M(l_k) \). 
Hence, in all cases, our claim holds: either \( s_i \in M(l_k) \) or \( l_k \) prefers \( s_i \) to the worst student in \( M(l_k) \). 
Finally, since \( s_i \) is either unassigned in both \( M \) and \( M' \) or prefers \( p_j \) to both $p_a$ and $p_b$, \( p_j \) is undersubscribed in \( M \), and either $s_i \in M(l_k)$ or $l_k$ prefers $s_i$ tot he worst student in $M(l_k)$, the pair \( (s_i,p_j) \) blocks \( M \), a contradiction.

\vspace{0.1cm}
\noindent \textbf{Case (S3 \& P3):} 
Here \( s_i \) is assigned in both \( M \) and \( M' \), prefers \( p_a \) to \( p_j \), and \( p_j \) to \( p_b \). By a similar argument to Cases (S1 \& P3) and (S2 \& P3), one of two situations arises in the construction of $M'$. If \( M(l_k) = M'(l_k) \), then \( p_j \) is undersubscribed in \( M \), \( l_k \) is full in \( M \), and \( s_i \in M(l_k) \). Since \( s_i \in M'(l_k) \) as well, \( l_k \) offers \( p_b \); however, by the construction of \( M' \), any time \( s_i \) is assigned to a different project of \( l_k \), the lecturer simultaneously loses a student from \( M(l_k) \), implying \( M(l_k) \neq M'(l_k) \), a contradiction. 
Hence this case cannot occur. 
Otherwise, \( M(l_k) \neq M'(l_k) \), and some student in \( M'(l_k) \) is preferred by \( l_k \) to the worst student in \( M(l_k) \). Since \( |M(p_j)| \le |M'(p_j)| \) and \( p_j \) is undersubscribed in \( M' \), it follows that \( p_j \) is also undersubscribed in \( M \). Also, by the construction of $M'$, $|M(l_k)| = |M'(l_k)|$, and \( l_k \) is full in \( M \). Moreover, \( l_k \) prefers \( s_i \) to the worst student in \( M(l_k) \).
Since \( p_j \) is undersubscribed in \( M \) and \( l_k \) prefers \( s_i \) to the worst student in \( M(l_k) \), it follows from the definition of \( s_M(s_i) \) that \( p_j \) should be a valid \( next_M(s_i) \), implying \( M'(s_i) = p_j \), a contradiction.

\vspace{0.2cm}
\noindent \textbf{Cases (S1 \& P4) and (S2 \& P4):} Based on (P4), it follows that \( p_j \) is undersubscribed in \( M \),  \( l_k \) is full in \( M \), and \( l_k \) prefers \( s_i \) to the worst student assigned in \( M(l_k) \). Specifically, if \( M(l_k) = M'(l_k) \), then \( p_j \) is undersubscribed in \( M \), \( l_k \) is full in \( M \), and \( l_k \) prefers \( s_i \) to the worst student in \( M(l_k) \). Alternatively, if \( M(l_k) \neq M'(l_k) \), then there exists some student \( s \in M'(l_k) \) such that \( l_k \) prefers \( s \) to the worst student in \( M(l_k) \), which implies that \( l_k \) also prefers \( s_i \) to the worst student in \( M(l_k) \). Hence our claim holds. Now consider $s_i$ who is either unassigned in both $M$ and $M'$, or prefers $p_j$ to $p_a$ and $p_b$. Since $p_j$ is undersubscribed in $M$ and $l_k$ prefers $s_i$ to the worst student in $M(l_k)$, it follows that $(s_i,p_j)$ blocks $M$, a contradiction. 

\vspace{0.2cm}
\noindent \textbf{Case (S3 \& P4):}
In this case, \( s_i \) prefers \( p_a \) to \( p_j \) and prefers \( p_j \) to \( p_b \). By applying a similar argument as in Cases (S1 \& P4) and (S2 \& P4), we conclude that $p_j$ is undersubscribed in $M$, $l_k$ is full in $M$, and \( l_k \) prefers \( s_i \) to the worst student in \( M(l_k) \). Now since \( p_j \) is undersubscribed in \( M \) and \( l_k \) prefers \( s_i \) to the worst student in \( M(l_k) \), it follows from the definition of \( s_M(s_i) \) that \( p_j \) must be a valid \( next_M(s_i) \), that is, \( M'(s_i) \) should be \( p_j \). This leads to a contradiction. 

\vspace{0.2cm}
\noindent We have now considered all possible conditions for the pair \((s_i,p_j)\) in $M'$, each resulting in a contradiction. Hence, \(M'\) is stable. Since every student in \(\rho\) receives a less preferred project in \(M'\) compared to \(M\), and all other students retain the same projects that they had in \(M\), it follows that \(M\) dominates \(M'\), that is, \(M\) dominates \(M/\rho\). This completes the proof.
\end{proof}

\subsection{Meta-rotations and stable matchings}
\noindent In this section, we show that every stable matching in a given {\sc spa-s} instance can be obtained by eliminating a specific set of meta-rotations starting from the student-optimal stable matching. This leads naturally to the definition of the meta-rotation poset in the next section. The main result in this section is Lemma~\ref{prop:s-in-rho}, where we prove that if \( \rho \) is exposed in stable matching \( M \), and some student \( s \in \rho \) prefers \( M \) to \( M' \), then every student in \( \rho \) prefers \( M \) to \( M' \). Moreover, if \( M \) dominates \( M' \), then either \( M' \) is the stable matching obtained by eliminating \( \rho \) from \( M \), that is, \( M' = M / \rho \), or \( M / \rho \) dominates \( M' \). This result is established using Lemmas~\ref{lem:s_M(s_i)-is-M'(s_i)} to \ref{lem:case2-s-in-rho}. 

\medskip
A key consequence of Lemmas \ref{lem:M-rho-is-stable} and ~\ref{prop:s-in-rho} is that it provides a systematic way to construct all stable matchings in a given instance, starting from the student-optimal matching. By successively eliminating an exposed meta-rotation, each step produces a new stable matching in which the students involved in the eliminated meta-rotation are assigned to projects they prefer less to their project in the previous matching. In this way, every stable matching can be reached through a sequence of such eliminations.

\begin{lemma}
\label{prop:s-in-rho}
Let \( M \) and \( M' \) be two stable matchings in a given {\sc spa-s} instance, and let \( \rho \) be a meta-rotation exposed in \( M \). Suppose there exists a student \( s_i \in \rho \) who prefers \( M \) to \( M' \). Then every student \( s \in \rho \) prefers \( M \) to \( M' \). Moreover, if \( M \) dominates \( M' \), then either \( M' \) is the stable matching obtained by eliminating \( \rho \) from \( M \), that is, \( M' = M / \rho \), or \( M / \rho \) dominates \( M' \).
\end{lemma}

\begin{proof}
\noindent Let \( M \) and \( M' \) be two stable matchings in \( I \), and let \( \rho \) be a meta-rotation exposed in \( M \). Suppose there exists a student \( s_i \in \rho \) who prefers \( M \) to \( M' \). Clearly, \( M(s_i) \neq M'(s_i) \), and \( s_M(s_i) \) exists. Moreover, \( s_i \) prefers \( M(s_i) \) to \( s_M(s_i) \). By Lemma~\ref{lem:no-p-between-M(s_i)-and-sM(s_i)}, there are no projects between $M(s_i)$ and $s_M(s_i)$ that form a stable pair with $s_i$. Therefore, either \( s_M(s_i) = M'(s_i) \), or \( s_i \) prefers \( s_M(s_i) \) to \( M'(s_i) \). Let \( p_j = s_M(s_i) \) where $l_k$ offers \( p_j \). By Definition \ref{def:exposed-mr}, there exists a student \( next_M(s_i) \) in \( \rho \), which we denote by \( s_z \).  Since \( s_z \in \rho \), \( s_M(s_z) \) exists, and \( s_z \) prefers \( M(s_z) \) to \( s_M(s_z) \). By the definition of \( next_M(s_i) \) (see Definition~\ref{def:next-project}), there are two possible conditions on $p_j$:
\begin{enumerate}[label=(\roman*)]
    \item \( p_j \) is full in \( M \), and \( s_z \) is the worst student in \( M(p_j) \), or
    \item \( p_j \) is undersubscribed in \( M \), \( l_k \) is full in \( M \), and \( s_z \) is the worst student in \( M(l_k) \).
\end{enumerate}
In both cases (i) and (ii), \( l_k \) prefers \( s_i \) to \( s_z \). 

\medskip
\noindent To prove Lemma~\ref{prop:s-in-rho}, it suffices to show that \( s_z \) also prefers \( M \) to \( M' \). Once this is established, the same reasoning can be extended to all other students in \( \rho \). To complete the proof, we make use of several auxiliary lemmas. Specifically, Lemma~\ref{lem:s_M(s_i)-is-M'(s_i)} covers the case where \( s_M(s_i) = M'(s_i) \), while Lemmas~\ref{lem:case1-s-in-rho} and~\ref{lem:case2-s-in-rho} address the case where \( s_i \) prefers \( s_M(s_i) \) to \( M'(s_i) \). In both Lemmas~\ref{lem:case1-s-in-rho} and~\ref{lem:case2-s-in-rho}, we first show that \( s_z \) is assigned to different projects in \( M \) and \( M' \), i.e., \( M(s_z) \neq M'(s_z) \), and then prove, by contradiction, that \( s_z \) prefers \( M \) to \( M' \). Together, these results establish Lemma~\ref{prop:s-in-rho}.
\end{proof}

\begin{lemma}
\label{lem:s_M(s_i)-is-M'(s_i)}
Let \( \rho \) be an exposed meta-rotation in \( M \), and suppose there exists a student \( s_i \in \rho \) who prefers $M$ to $M'$ and $s_M(s_i) = M'(s_i)$. If \( s_i \) prefers \( M \) to \( M' \), then \( s_z \) prefers $M$ to $M'$.
\end{lemma}

\begin{proof}
\noindent
Let \( s_i \in \rho \) be some student who prefers \( M \) to \( M' \), and suppose that $s_M(s_i) = M'(s_i)$. This implies that \( M' \) is the stable matching obtained by eliminating \( \rho \) from \( M \). Moreover,  by Lemma~\ref{lem:M-rho-is-stable}, \( M \) dominates \( M' \). Recall that \( p_j = s_M(s_i) \); thus, \( s_i \in M'(p_j) \setminus M(p_j) \). Since \( s_i \) is assigned to \( p_j \) in \( M' \), it follows from Lemma~\ref{lem:swaps} that, regardless of whether \( p_j \) is full or undersubscribed in \( M \), the worst student in \( M(p_j) \) or \( M(l_k) \), denoted \( s_z \), must be assigned to a different project in \( M \) and \( M' \). In particular, \( s_z \in M(p_j) \setminus M'(p_j) \). Moreover, since \( M \) dominates \( M' \), it follows that \( s_z \) prefers \( M \) to \( M' \). This completes the proof.
\end{proof}

\begin{lemma}
\label{lem:case1-s-in-rho}
Let \( \rho \) be an exposed meta-rotation in \( M \), where \( (s_i, p_j) \in \rho \) and \( s_i \) prefers \( p_j \) to \( M'(s_i) \). If \( p_j \) is full in \( M \) and \( s_z \) is the worst student in \( M(p_j) \), then \( s_z \) prefers \( M \) to \( M' \).
\end{lemma}

\begin{proof}
\noindent Let \( M \) be a stable matching in which \( \rho \) is exposed, and suppose that some student \( s_i \in \rho \) prefers \( M \) to \( M' \). Let $s_z \in \rho$ be the worst student in $M(p_j)$. We note that $l_k$ prefers $s_i$ to $s_z$. First suppose for a contradiction that \( M(s_z) = M'(s_z) \).  Then, regardless of whether \( p_j \) is full or undersubscribed in \( M' \), the pair \( (s_i, p_j) \) blocks \( M' \), since \( s_i \) prefers \( p_j \) to \( M'(s_i) \), and \( l_k \) prefers \( s_i \) to some student in \( M'(p_j) \) (namely \( s_z \)). This contradicts the stability of $M'$. Hence, \( M(s_z) \neq M'(s_z) \). Now, suppose for a contradiction that $s_z$ prefers $M'$ to $M$, that is, \( s_z \) prefers \( M'(s_z) \) to \( p_j \). We consider cases (A) and (B), depending on whether \( p_j \) is full or undersubscribed in \( M' \).

\medskip
\noindent \textbf{(A):} \( p_j \) is full in \( M' \). Since $p_j$ is also full in $M$, there exists some student \( s_a \in M'(p_j) \setminus M(p_j) \). By Lemma~\ref{lem:chap4-s-prefers-l-prefers1}, since $s_z$ prefers $M'(s_z)$ to $p_j$, \( l_k \) prefers \( s_z \) to each student in \( M'(p_j) \setminus M(p_j) \), so \( l_k \) prefers \( s_z \) to \( s_a \). Additionally, since \( s_i \) prefers \( p_j \) to \( M'(s_i) \) and \( p_j \) is full in $M'$, \( l_k \) prefers each student in \( M'(p_j) \) to \( s_i \), implying \( l_k \) prefers \( s_a \) to \( s_i \). Since \( l_k \) prefers \( s_z \) to \( s_a \), and prefers \( s_a \) to \( s_i \), it follows that \( l_k \) prefers \( s_z \) to \( s_i \). However, by definition of $s_M(s_i)$, \( l_k \) prefers \( s_i \) to \( s_z \), which yields a contradiction. Therefore, our claim holds and \( s_z \) prefers $M$ to $M'$.

\medskip
\noindent \textbf{(B):} \( p_j \) is undersubscribed in \( M' \). By Lemma~\ref{lem:chap4-s-prefers-l-prefers1}, since \( s_z \) prefers \( M'(s_z) \) to \( p_j \), \( l_k \) prefers \( s_z \) to each student in \( M'(l_k) \setminus M(l_k) \). Moreover, if \( s_z \in S_k(M, M') \), then by Lemma~\ref{lem:chap4samelecturer1}, there exists at least one student in \( M(l_k) \setminus M'(l_k) \) who $l_k$ prefers to $s_z$, or we have \( s_z \in M(l_k) \setminus M'(l_k) \) itself. Consequently, it follows that there also exists a student in \( M'(l_k) \setminus M(l_k) \). Let \( s_b \) denote the worst student in \( M'(l_k) \setminus M(l_k) \). Then \( l_k \) prefers \( s_z \) to \( s_b \). Since \( s_i \) prefers \( p_j \) to \( M'(s_i) \), and \( p_j \) is undersubscribed in \( M' \), \( l_k \) prefers each student in \( M'(l_k) \) (including \( s_b \)) to \( s_i \). Since \( l_k \) prefers \( s_z \) to \( s_b \), and prefers \( s_b \) to \( s_i \), it follows that \( l_k \) prefers \( s_z \) to \( s_i \); This again contradicts the assumption that \( l_k \) prefers \( s_i \) to \( s_z \) (by definition of $next_M(s_i)$). Hence, \( s_z \) prefers $M$ to $M'$, and our claim holds.
\end{proof}

\begin{lemma}
\label{lem:case2-s-in-rho}
Let \( \rho \) be an exposed meta-rotation in \( M \), where \( (s_i, p_j) \in \rho \) and \( s_i \) prefers \( p_j \) to \( M'(s_i) \). If \( p_j \) is undersubscribed in \( M \) and \( s_z \) is the worst student in \( M(l_k) \), then \( s_z \) prefers \( M \) to \( M' \).
\end{lemma}

\begin{proof}
Let \( M \) be a stable matching in which \( \rho \) is exposed, and suppose that some student \( s_i \in \rho \) prefers \( M \) to \( M' \). Let \( s_z \in \rho \) be the worst student in \( M(l_k) \). Note that, by definition of $s_M(s_i)$, \( l_k \) prefers \( s_i \) to \( s_z \). We first show, in case (A), that \( s_z \) is assigned to different lecturers in \( M \) and \( M' \). We then show, in case (B), that \( s_z \) prefers \( M \) to \( M' \).

\medskip
\noindent \textbf{(A):} Suppose for a contradiction that \( s_z \in M(l_k) \cap M'(l_k) \). We consider subcases (A1) and (A2) depending on whether $p_j$ is full or undersubscribed in $M'$.

\medskip
\noindent \textbf{(A1):} \( p_j \) is full in \( M' \). Since \( p_j \) is undersubscribed in \( M \), there exists a student \( s_a \in M'(p_j) \setminus M(p_j) \). Since \( s_i \) prefers \( p_j \) to \( M'(s_i) \) and \( p_j \) is full in \( M' \), it follows that \( l_k \) prefers each student in \( M'(p_j) \) to \( s_i \). Therefore, \( l_k \) prefers \( s_a \) to \( s_i \). If \( s_a \) prefers \( p_j \) to \( M(s_a) \), then since \( p_j \) is undersubscribed in \( M \), \( l_k \) prefers each student in \( M(l_k) \) to \( s_a \). In particular, \( l_k \) prefers \( s_z \) to \( s_a \), since $s_z \in M(l_k)$. Furthermore, since $l_k$ prefers $s_z$ to $s_a$, and prefers $s_a$ to $s_i$, it follows that \( l_k \) prefers \( s_z \) to \( s_i \); this contradicts the fact that \( l_k \) prefers \( s_i \) to \( s_z \). Therefore, \( s_a \) prefers \( M(s_a) \) to \( p_j \). Moreover, by Lemma~\ref{lem:chap4-s-prefers-l-prefers1}, since \( p_j \) is undersubscribed in \( M \), \( l_k \) prefers \( s_a \) to each student in \( M(l_k) \setminus M'(l_k) \).

\medskip
\noindent Now, since \( |M'(p_j)| > |M(p_j)| \) and \( |M(l_k)| = |M'(l_k)| \), there exists some project \( p_b \in P_k \) such that \( |M(p_b)| > |M'(p_b)| \). This implies there exists a student \( s_b \in M(p_b) \setminus M'(p_b) \), and \( p_b \) is undersubscribed in \( M' \). Moreover, \( l_k \) prefers \( s_b \) to \( s_z \), since \( s_b \in M(l_k) \) and \( s_z \) is the worst student in \( M(l_k) \). If \( s_b \) prefers \( p_b \) to \( M'(s_b) \), then since \( p_b \) is undersubscribed in \( M' \), \( l_k \) prefers each student in \( M'(l_k) \) to \( s_b \). In particular, \( l_k \) prefers \( s_z \) (who is also in $M'(l_k)$) to \( s_b \), contradicting the earlier fact that \( l_k \) prefers \( s_b \) to \( s_z \). Therefore, \( s_b \) prefers \( M'(s_b) \) to \( p_b \). By Lemma~\ref{lem:chap4-s-prefers-l-prefers1} (applied with \( M \) and \( M' \) swapped), since \( p_b \) is undersubscribed in \( M' \), \( l_k \) prefers \( s_b \) to each student in \( M'(l_k) \setminus M(l_k) \).

\medskip
\noindent We now show that the combination of conditions where \( s_a \) prefers \( M \) to \( M' \) and \( l_k \) prefers \( s_a \) to each student in \( M(l_k) \setminus M'(l_k) \), together with the conditions where \( s_b \) prefers \( M' \) to \( M \) and \( l_k \) prefers \( s_b \) to each student in \( M'(l_k) \setminus M(l_k) \), leads to a contradiction.  Suppose \( s_a \in M'(l_k) \setminus M(l_k) \). Then \( l_k \) prefers \( s_b \) to \( s_a \), since \( l_k \) prefers \( s_b \) to each student in \( M'(l_k) \setminus M(l_k) \). Next, suppose \( s_a \in S_k(M, M') \). By Lemma~\ref{lem:chap4samelecturer1}, since $s_a$ prefers $M$ to $M'$, then there exists some student \( s_r \in M'(l_k) \setminus M(l_k) \) such that \( l_k \) prefers \( s_r \) to \( s_a \). Given that \( l_k \) prefers \( s_b \) to each student in \( M'(l_k) \setminus M(l_k) \), it follows that \( l_k \) prefers \( s_b \) to \( s_r \), and thus \( l_k \) prefers \( s_b \) to \( s_a \).

\medskip
\noindent A similar argument applies to \( s_b \). Suppose \( s_b \in M(l_k) \setminus M'(l_k) \). Then \( l_k \) prefers \( s_a \) to \( s_b \), since \( l_k \) prefers \( s_a \) to each student in \( M(l_k) \setminus M'(l_k) \). On the other hand, suppose \( s_b \in S_k(M, M') \). By Lemma~\ref{lem:chap4samelecturer1} (applied with \( M \) and \( M' \) swapped), there exists a student \( s_r \in M(l_k) \setminus M'(l_k) \) such that \( l_k \) prefers \( s_r \) to \( s_b \). Moreover, since \( l_k \) prefers \( s_a \) to each student in \( M(l_k) \setminus M'(l_k) \), it follows that \( l_k \) prefers \( s_a \) to \( s_r \), and thus \( l_k \) prefers \( s_a \) to \( s_b \). This yields a contradiction since \( l_k \) cannot simultaneously prefer \( s_b \) to \( s_a \) and \( s_a \) to \( s_b \). Therefore, the conditions under which \( s_a \) prefers \( M \) to \( M' \), while \( s_b \) prefers \( M' \) to \( M \), result in a contradiction on the preferences of \( l_k \). Hence, $s_z \in M(l_k) \setminus M'(l_k)$, and this completes the proof for (A1).

\medskip
\noindent \textbf{(A2):} \( p_j \) is undersubscribed in \( M' \). Since \( s_i \) prefers \( p_j \) to \( M'(s_i) \), it follows that \( l_k \) prefers each student in \( M'(l_k) \) to \( s_i \). If \( s_z \in M'(l_k) \), then \( l_k \) prefers \( s_z \) to \( s_i \), which directly contradicts the assumption that \( l_k \) prefers \( s_i \) to \( s_z \). Hence, \( s_z \in M(l_k) \setminus M'(l_k)\).

\medskip
\noindent We now show in case (B) that \( s_z \) prefers \( M \) to \( M' \), given that \( M(s_z) \neq M'(s_z) \).

\medskip \noindent \textbf{(B):}  Suppose for a contradiction that \( s_z \) prefers \( M' \) to \( M \). Again, we consider subcases (B1) and (B2)  depending on whether $p_j$ is full or undersubscribed in $M'$.

\medskip
\noindent \textbf{(B1):} \( p_j \) is full in \( M' \). Similar to case (A1), we show that we can identify a student in \( M'(l_k) \setminus M(l_k) \) who prefers \( M \) to \( M' \), and a student in \( M(l_k) \setminus M'(l_k) \) who prefers \( M' \) to \( M \), which yields a contradiction based on $l_k$'s preferences. Since \( |M'(p_j)| > |M(p_j)| \), there exists a student \( s_a \in M'(p_j) \setminus M(p_j) \). Given that \( s_i \) prefers \( p_j \) to \( M'(s_i) \) and \( p_j \) is full in \( M' \), it follows that \( l_k \) prefers \( s_a \) to \( s_i \). We also know that \( l_k \) prefers \( s_i \) to \( s_z \), with \( s_z \in M(l_k) \). Therefore, \( l_k \) prefers \( s_a \) to $s_z$. Now, if \( s_a \) prefers \( M' \) to \( M \), then \( p_j \) is undersubscribed in \( M \), and \( l_k \) would the worst student in $M(l_k)$ (namely $s_z$) to \( s_a \), which yields a contradiction to the fact that $l_k$ prefers $s_a$ to $s_z$. Thus, \( s_a \) prefers \( M \) to \( M' \). In particular, this implies that \( s_a \) prefers \( M(s_a) \) to \( p_j \), \( p_j \) is undersubscribed in \( M \), and by Lemma~\ref{lem:chap4samelecturer1}, \( l_k \) prefers \( s_a \) to each student in \( M(l_k) \setminus M'(l_k) \).

\medskip
\noindent Recall that $s_z \in M(l_k) \setminus M'(l_k)$ and prefers $M'$ to $M$. Let $M(s_z)$ be $p_z$, where $p_z \in P_k$. Let $s_{z'}$ be the worst student in $M'(l_k)$. Since $s_z$ prefers $M'(s_z)$ to $p_z$, whether $p_z$ is full or undersubscribed in $M'$, it follows that $l_k$ prefers $s_z$ to the worst student in $M'(l_k)$. Therefore $l_k$ prefers $s_z$ to $s_{z'}$.

\medskip
\noindent Since \( |M'(p_j)| > |M(p_j)| \) and \( |M(l_k)| = |M'(l_k)| \), there exists a project \( p_b \in P_k \) such that \( |M(p_b)| > |M'(p_b)| \). This implies that there exists a student \( s_b \in M(p_b) \setminus M'(p_b) \), and \( p_b \) is undersubscribed in \( M' \). Moreover, \( l_k \) prefers \( s_b \) to \( s_z \), since \( s_b \in M(l_k) \) and \( s_z \) is the worst student in \( M(l_k) \). If \( s_b \) prefers \( p_b \) to \( M'(s_b) \), then, because \( p_b \) is undersubscribed in \( M' \), it follows that \( l_k \) prefers each student in \( M'(l_k) \) to \( s_b \). In particular, \( l_k \) prefers \( s_{z'} \), the worst student in \( M'(l_k) \), to \( s_b \). Additionally, since \( l_k \) prefers \( s_z \) to \( s_{z'} \), it follows that \( l_k \) prefers \( s_z \) to \( s_b \). However, this contradicts the fact that \( s_z \) is the worst student in \( M(l_k) \), since it implies that \( l_k \) prefers \( s_z \) to another student \( s_b \) who is also assigned to \( M(l_k) \).
Therefore, we conclude that \( s_b \) prefers \( M'(s_b) \) to \( p_b \). By Lemma~\ref{lem:chap4-s-prefers-l-prefers1} (applied with \( M \) and \( M' \) swapped), since \( p_b \) is undersubscribed in \( M' \), it follows that \( l_k \) prefers \( s_b \) to each student in \( M'(l_k) \setminus M(l_k) \).

\medskip
\noindent We now show that combining the conditions where \( s_a \) prefers \( M \) to \( M' \) and \( l_k \) prefers \( s_a \) to every student in \( M(l_k) \setminus M'(l_k) \), together with the conditions where \( s_b \) prefers \( M' \) to \( M \) and \( l_k \) prefers \( s_b \) to every student in \( M'(l_k) \setminus M(l_k) \), leads to a contradiction.

\medskip
\noindent First suppose \( s_a \in M'(l_k) \setminus M(l_k) \). Then \( l_k \) prefers \( s_b \) to \( s_a \), since \( l_k \) prefers \( s_b \) to each student in \( M'(l_k) \setminus M(l_k) \). Next, suppose \( s_a \in S_k(M, M') \) where $s_a$ prefers $M$ to $M'$. By Lemma~\ref{lem:chap4samelecturer1}, there exists a student \( s_r \in M'(l_k) \setminus M(l_k) \) such that \( l_k \) prefers \( s_r \) to \( s_a \). Since \( l_k \) prefers \( s_b \) to each student in \( M'(l_k) \setminus M(l_k) \), it follows that \( l_k \) prefers \( s_b \) to \( s_r \), and thus \( l_k \) prefers \( s_b \) to \( s_a \).

\medskip
\noindent A similar argument applies to \( s_b \). Suppose \( s_b \in M(l_k) \setminus M'(l_k) \). Then \( l_k \) prefers \( s_a \) to \( s_b \), since \( l_k \) prefers \( s_a \) to each student in \( M(l_k) \setminus M'(l_k) \). On the other hand, suppose \( s_b \in S_k(M, M') \) where $s_b$ prefers $M'$ to $M$. By Lemma~\ref{lem:chap4samelecturer1}, there exists a student \( s_r \in M(l_k) \setminus M'(l_k) \) such that \( l_k \) prefers \( s_r \) to \( s_b \). Moreover, since \( l_k \) prefers \( s_a \) to each student in \( M(l_k) \setminus M'(l_k) \), it follows that \( l_k \) prefers \( s_a \) to \( s_r \), and thus \( l_k \) prefers \( s_a \) to \( s_b \). In both cases, we reach a contradiction, since \( l_k \) cannot simultaneously prefer \( s_b \) to \( s_a \) and \( s_a \) to \( s_b \). Therefore, $s_z$ prefers $M$ to $M'$, and this completes the proof.

\medskip \noindent \textbf{(B2):} \( p_j \) is undersubscribed in \( M' \). Since \( s_z \in M(l_k) \setminus M'(l_k) \), there exists some student \( s_{z'} \in M'(l_k) \setminus M(l_k) \). Since \( s_i \) prefers \( p_j \) to \( M'(s_i) \) and $p_j$ is undersubscribed in $M'$, it follows that \( l_k \) prefers each student in \( M'(l_k) \) to \( s_i \). In particular, \( l_k \) prefers \( s_{z'} \) to \( s_i \). Recall that \( s_z \) prefers \( M' \) to \( M \); let \( p_z = M(s_z) \). Whether \( p_z \) is full or undersubscribed in \( M' \), it follows from Lemma~\ref{lem:chap4-s-prefers-l-prefers1} that \( l_k \) prefers \( s_z \) to each student in \( M'(l_k) \setminus M(l_k) \). In particular, \( l_k \) prefers \( s_z \) to \( s_{z'} \).  Combining these observations, we have that \( l_k \) prefers \( s_z \) to \( s_{z'} \), and \( s_{z'} \) to \( s_i \), which implies that \( l_k \) prefers \( s_z \) to \( s_i \). This contradicts the assumption that \( l_k \) prefers \( s_i \) to \( s_z \). Hence, we conclude that \( s_z \) prefers \( M \) to \( M' \). Therefore, $s_z$ prefers $M$ to $M'$, and this completes the proof for case (B2).

\medskip
\noindent Thus, in both cases (B1) and (B2), $s_z$ prefers $M$ to $M'$. This completes the proof.
\end{proof}

\noindent The arguments in Lemmas~\ref{lem:case1-s-in-rho} and \ref{lem:case2-s-in-rho} can be extended to every student in \( \rho \), since by Definitions~\ref{def:next-project} and \ref{def:exposed-mr}, each student in \( \rho \) has a valid next student who is also in \( \rho \). Therefore, if \( s_i \in \rho \) prefers \( M \) to \( M' \), then every student \( s \in \rho \) also prefers \( M \) to \( M' \).  

\medskip
\noindent Now, suppose that \( M \) dominates \( M' \). By Lemma~\ref{lem:no-p-between-M(s_i)-and-sM(s_i)}, for each student $s_i \in \rho$, there is no stable pair that lies between their assigned projects in $M$ and $M/\rho$. Hence, it follows that \( M / \rho \) either dominates \( M' \) or is equal to \( M' \), since only the students in \( \rho \) have different projects in \( M \) and \( M / \rho \). Moreover, each of these students prefers \( M \) to \( M' \), with the possibility that \( M / \rho = M' \).  This completes the proof of Lemma~\ref{prop:s-in-rho}. In addition, this lemma immediately implies Corollary~\ref{cor:s-in-rho-all-s-in-rho}.

\medskip
\begin{corollary}
\label{cor:s-in-rho-all-s-in-rho}
Let 
\(\rho = \{(s_{0}, p_{0}), (s_{1}, p_{1}), \ldots, (s_{r-1}, p_{r-1})\}\)
be a meta-rotation of \( I \). 
If there exists a stable matching \( M' \) such that, for some pair \( (s_a, p_a) \in \rho \), student \( s_a \) prefers \( p_a \) to their project in \( M' \), then for every \( t \in \{0, \ldots, r-1\} \), student \( s_t \) prefers \( p_t \) to \( M'(s_t) \).
\end{corollary}

\medskip
\noindent In the following subsections, we describe a pruning step and a method for obtaining a target stable matching using meta-rotations.

\vspace{-0.1cm}
\subsubsection{Pruning step}
\label{sect:initial-pruning}
\noindent
We construct a reduced instance \( \hat{I} \) from a given {\sc spa-s} instance \( I \) as follows. 
First, apply the student-oriented algorithm to obtain the student-optimal stable matching \( M_S \) and remove all pairs that cannot appear in any stable matching. Then, apply the lecturer-oriented algorithm to compute the lecturer-optimal stable matching \( M_L \) and eliminate additional non-stable pairs. The resulting instance after both steps is the reduced instance \( \hat{I} \).

\vspace{-0.1cm}
\subsubsection{Finding a target stable matching}
\noindent
Any target stable matching in a given instance can be obtained from the student-optimal matching by successively exposing and eliminating meta-rotations. Given a {\sc spa-s} instance \( I \) and a target stable matching \( M_T \), apply the pruning step above to obtain the reduced instance \( \hat{I} \) with student-optimal matching \( M \). 
If \( M = M_T \), we are done. Otherwise, since \( M \) dominates \( M_T \), there exists a student \( s \) such that \( M(s) \neq M_T(s) \) and \( s \) prefers \( M \) to \( M_T \). 
By Lemma~\ref{lem:one-metarotation-in-M}, \( M \) has an exposed meta-rotation \( \rho \) starting at \( s \); eliminating it yields a stable matching \( M / \rho \) by Lemma~\ref{lem:M-rho-is-stable}. 
By Lemma~\ref{prop:s-in-rho}, either \( M / \rho = M_T \) or \( M / \rho \) dominates \( M_T \). 
Repeating this process i.e. identifying the exposed meta-rotation starting at a student whose project differs between the current matching and \( M_T \), and eliminating it, eventually yields \( M_T \).

\medskip
\noindent \textbf{Example:} Here, we illustrate how to identify all exposed meta-rotations and describe the transitions between stable matchings using the {\sc spa-s} instance \( I_1 \), shown in Figure~\ref{fig:instance1}. We begin by constructing the reduced instance corresponding to \( I_1 \), following the steps outlined in Section~\ref{sect:initial-pruning}. From Table~\ref{tab:instance1}, we observe that \( M_7 \) is the lecturer-optimal stable matching for \( I_1 \). In $M_7$, student \( s_1 \) is assigned to project \( p_4 \), which is the worst project they are assigned to in any stable matching. Consequently, we remove all projects that are less preferred than \( p_4 \) from \( s_1 \)'s preference list. Here, project \( p_3 \) is deleted from $s_1$'s list. Continuing this pruning process for all students yields the reduced instance for instance \( I_1 \), which is presented in Figure~\ref{fig:example1reduced}.

\begin{figure}[h]
\centering
\renewcommand{\arraystretch}{1}
\setlength{\tabcolsep}{1pt} 
\resizebox{\textwidth}{!}{ 
\begin{tabular}{p{0.3\textwidth} p{0.35\textwidth} p{0.35\textwidth}}
\hline
$s_1$: $p_1 \; p_2 \; p_4$ & $l_1$: $s_7 \; s_9 \; s_3 \; s_4 \; s_1 \; s_2 \; s_6 \; s_8 $ & $p_1$, $p_2$, $p_5$, $p_6$ \\ 

$s_2$: $p_1 \; p_4 \; p_3$ & $l_2$: $s_6 \; s_1 \; s_2 \; s_5 \; s_3 \; s_4 \; s_7 \; s_8 \;  s_9$ & $p_3$, $p_4$, $p_7$, $p_8$ \\ 

$s_3$: $p_3 \; p_1 \; p_2$ & & \\ 
$s_4$: $p_3 \; p_2 \; p_1$ & & \\ 
$s_5$: $p_4 \; p_3$ & & \\ 
$s_6$: $p_5 \; p_2 \; p_7$ & & \\ 
$s_7$: $p_7 \; p_3 \; p_6$ & & \\ 
$s_8$: $p_6 \; p_8$ & \multicolumn{2}{l}{\textbf{Project capacities:} $c_1 = c_3 = 2$; $\forall j \in \{2, 4, 5, 6, 7, 8\}, \, c_j = 1$} \\ 
$s_9$: $p_8 \; p_2 \;$ & \multicolumn{2}{l}{\textbf{Lecturer capacities:} $d_1 = 4$, $d_2 = 5$}  \\ 
\hline
\end{tabular}
}
\caption{Reduced preference list for $I_1$} 
\label{fig:example1reduced}
\end{figure}

\noindent Table~\ref{tab:m1} shows, for each student \( s_i \) in \( M_1 \), the next project \( p \) (denoted \( s_{M_1}(s_i) \)) and the student \( \mathrm{next}_{M_1}(s_i) \), defined as either the worst student in \( M_1(p) \) if \( p \) is full in \( M \), or the worst student in \( M_1(l_k) \) if \( p \) is undersubscribed in \( M \). As an illustration, consider \(s_1\): \(p_2\) is the first project after \(p_1\) such that \(p_2\) is undersubscribed in \(M_1\) and \(l_1\) (who offers \(p_1\)) prefers \(s_1\) to the worst student in \(M_1(l_1)\), namely \(s_8\). Consequently, \(\mathrm{next}_{M_1}(s_1) = s_8\). The remaining entries can be verified in a similar manner. We observe that the meta-rotation \( \rho_1 = \{(s_8,p_6), (s_9,p_8)\} \) is the only exposed meta-rotation in \(M_1\). Moreover, \(s_8\) is the worst student in \(p_6\) and \(\mathrm{next}_{M_1}(s_8) = s_9\). Likewise, \(s_9\) is the worst student in \(p_8\), and \(\mathrm{next}_{M_1}(s_9) = s_8\). Eliminating $\rho_1$ from $M_1$ gives $M_2$, that is, $M_1/\rho_1 = M_2$.\\

\begin{table}[h]
\centering
\resizebox{\textwidth}{!}{%
\begin{tabular}{|c|c|c|c|c|c|c|c|c|c|}
\hline
\textbf{$(s_i,p_j)$} & $(s_1,p_1)$ & $(s_2,p_1)$ & $(s_3,p_3)$ & $(s_4,p_3)$ & $(s_5,p_4)$ & $(s_6,p_5)$ & $(s_7,p_7)$ & $(s_8,p_6)$ & $(s_9,p_8)$ \\
\hline
\textbf{$s_{M_1}(s_i)$} & $p_2$ & $p_4$ & $p_1$ & $p_2$ & $p_3$ & $p_2$ & $p_6$ & $p_8$ & $p_2$ \\
\hline
\textbf{$next_{M_1}(s_i)$} & $s_8$ & $s_5$ & $s_2$ & $s_8$ & $s_4$ & $s_8$ & $s_8$ & $s_9$ & $s_8$ \\
\hline
\end{tabular}%
}
\caption{$s_{M_1}(s_i)$ and $next_{M_1}(s_i)$ for each student $s_i$ in $M_1$}
\label{tab:m1}
\end{table}

\noindent Similarly, Table \ref{tab:m2} shows \(s_{M_2}(s_i)\) and \(\mathrm{next}_{M_2}(s_i)\) for each student \(s_i\) in \(M_2\). In $M_2$, there are two exposed meta-rotations namely $\rho_2 = \{(s_6,p_5), (s_7,p_7)\}$ and $\rho_3 = \{(s_2,p_1), (s_5,p_4), (s_4,p_3)\}$. 
$M_2/\rho_2 = M_3$ and  $M_2/\rho_3 = M_4$.\\ 

\begin{table}[h]
\centering
\resizebox{\textwidth}{!}{%
\begin{tabular}{|c|c|c|c|c|c|c|c|c|c|}
\hline
\textbf{$(s_i,p_j)$} & $(s_1,p_1)$ & $(s_2,p_1)$ & $(s_3,p_3)$ & $(s_4,p_3)$ & $(s_5,p_4)$ & $(s_6,p_5)$ & $(s_7,p_7)$ & $(s_8,p_8)$ & $(s_9,p_2)$ \\
\hline
\textbf{$s_{M_2}(s_i)$} & $p_4$ & $p_4$ & $p_1$ & $p_1$ & $p_3$ & $p_7$ & $p_6$ & $-$ & $-$ \\
\hline
\textbf{$next_{M_2}(s_i)$} & $s_5$ & $s_5$ & $s_2$ & $s_2$ & $s_4$ & $s_7$ & $s_6$ & $-$ & $-$ \\
\hline
\end{tabular}%
}
\caption{$s_{M_2}(s_i)$ and $next_{M_2}(s_i)$ for each student $s_i$ in $M_2$}
\label{tab:m2}
\end{table}

\noindent Let $M_3$ be the next stable matching obtained by eliminating $\rho_2$ from $M_2$. Table \ref{tab:m3} shows \(s_{M_3}(s_i)\) and \(\mathrm{next}_{M_3}(s_i)\) for each student \(s_i\) in \(M_3\). In $M_3$, there is one exposed meta-rotation namely $\rho_3 = \{(s_2,p_1), (s_5,p_4), (s_4,p_3)\}$. Also, $M_3/\rho_3 = M_5$.\\

\begin{table}[h]
\centering
\resizebox{\textwidth}{!}{%
\begin{tabular}{|c|c|c|c|c|c|c|c|c|c|}
\hline
\textbf{$(s_i,p_j)$} & $(s_1,p_1)$ & $(s_2,p_1)$ & $(s_3,p_3)$ & $(s_4,p_3)$ & $(s_5,p_4)$ & $(s_6,p_7)$ & $(s_7,p_6)$ & $(s_8,p_8)$ & $(s_9,p_2)$ \\
\hline
\textbf{$s_{M_3}(s_i)$} & $p_4$ & $p_4$ & $p_1$ & $p_1$ & $p_3$ & $-$ & $-$ & $-$ & $-$ \\
\hline
\textbf{$next_{M_3}(s_i)$} & $s_5$ & $s_5$ & $s_2$ & $s_2$ & $s_4$ & $-$ & $-$ & $-$ & $-$ \\
\hline
\end{tabular}%
}
\caption{$s_{M_3}(s_i)$ and $next_{M_3}(s_i)$ for each student $s_i$ in $M_3$}
\label{tab:m3}
\end{table}

\noindent Table \ref{tab:m5} shows \(s_{M_5}(s_i)\) and \(\mathrm{next}_{M_5}(s_i)\) for each student \(s_i\) in \(M_5\). Clearly, the meta-rotation $\rho_4 = \{(s_1,p_1),(s_2,p_4),(s_3,p_3)\}$ is exposed in $M_5$, and $M_5/\rho_4 = M_7$.
\begin{table}[h]
\centering
\resizebox{\textwidth}{!}{%
\begin{tabular}{|c|c|c|c|c|c|c|c|c|c|}
\hline
\textbf{$(s_i,p_j)$} & $(s_1,p_1)$ & $(s_2,p_4)$ & $(s_3,p_3)$ & $(s_4,p_1)$ & $(s_5,p_3)$ & $(s_6,p_7)$ & $(s_7,p_6)$ & $(s_8,p_8)$ & $(s_9,p_2)$ \\
\hline
\textbf{$s_{M_5}(s_i)$} & $p_4$ & $p_3$ & $p_1$ & $-$ & $-$ & $-$ & $-$ & $-$ & $-$ \\
\hline
\textbf{$next_{M_5}(s_i)$} & $s_2$ & $s_3$ & $s_1$ & $-$ & $-$ & $-$ & $-$ & $-$ & $-$ \\
\hline
\end{tabular}%
}
\caption{$s_{M_5}(s_i)$ and $next_{M_5}(s_i)$ for each student $s_i$ in $M_5$}
\label{tab:m5}
\end{table}

\medskip
\noindent We have identified a total of four meta-rotations in instance \( I_1 \): \(\rho_1\), \(\rho_2\), \(\rho_3\), and \(\rho_4\), each of which is exposed in at least one stable matching of \( I_1 \). We also observe that a meta-rotation can be exposed in multiple stable matchings, and that a single stable matching may contain more than one exposed meta-rotation. For example, the meta-rotation \(\rho_2 = \{(s_6, p_5), (s_7, p_7)\}\) is exposed in \( M_2 \), \( M_4 \), and \( M_6 \). Furthermore, the stable matching \( M_2 \) contains both \(\rho_2\) and \(\rho_3\) as exposed meta-rotations.

\section{Meta-rotation poset}
\label{sect:meta-rotationposet}
In this section, we show that for any {\sc spa-s} instance \( I \), we can define a partial order on its set of meta-rotations, forming a partially ordered set (poset), such that each stable matching corresponds to a unique closed subset of this poset. Given a {\sc spa-s} instance \( I \), let \( \mathcal{M} \) denote the set of stable matchings in \( I \), and let \( R \) be the set of meta-rotations that are exposed in some stable matching in \( \mathcal{M} \). For any two meta-rotations \( \rho_1, \rho_2 \in R \), we define a relation \( \prec \) such that \( \rho_1 \prec \rho_2 \) if every stable matching in which \( \rho_2 \) is exposed can be obtained only after $\rho_1$ has been eliminated, and there is no other meta-rotation \( \rho' \in R \setminus \{ \rho_1, \rho_2 \} \) such that \( \rho_1 \prec \rho' \prec \rho_2 \). In this case, we say that \( \rho_1 \) is an \emph{immediate predecessor} of \( \rho_2 \).

\begin{definition}[Meta-rotation poset]
\label{def:metarotation-poset}
Let \( R \) be the set of meta-rotations in a {\sc spa-s} instance \( I \), and let \( \prec \) be the immediate predecessor relation on \( R \). We define a relation \( \leq \) on \( R \) such that \( \rho_1 \leq \rho_2 \) if and only if either \( \rho_1 = \rho_2 \), or there exists a finite sequence of meta-rotations \( \rho_1 \prec \rho_{u} \prec \cdots \prec \rho_{v} \prec \rho_2 \). The pair \( (R, \leq) \) is called the \emph{meta-rotation poset} for instance \( I \). 
\end{definition}

\begin{proposition}
\label{proposition-metarotation-poset}
Let \( R \) be the set of meta-rotations in a given {\sc spa-s} instance \( I \), and let \( \leq \) be the relation on \( R \) defined as above. Then \( (R, \leq) \) is a partially ordered set.
\end{proposition}

\begin{proof}
We will show that the relation \( \leq \) on \( R \) is (i) reflexive, (ii) antisymmetric, and (iii) transitive.

\begin{enumerate}[label = (\roman*)]
    \item \textbf{Reflexivity:} Let \( \rho \in R \). By definition, every element is related to itself. Hence, \( \rho \leq \rho \), and \( \leq \) is reflexive.

    \item \textbf{Antisymmetry:} Suppose there exist \( \rho_1, \rho_2 \in R \) such that \( \rho_1 \leq \rho_2 \) and \( \rho_2 \leq \rho_1 \). We claim that \( \rho_1 = \rho_2 \). Suppose, for contradiction, that \( \rho_1 \neq \rho_2 \). By the definition of \( \leq \), there exists a sequence of meta-rotation eliminations \( \rho_1 \prec \rho_{u} \prec \cdots \prec \rho_2 \), and another sequence \( \rho_2 \prec \rho_{v} \prec \cdots \prec \rho_1 \). Now, consider any stable matching in which \( \rho_1 \) is exposed. From the second sequence, we conclude that \( \rho_2 \) must have been eliminated before \( \rho_1 \) can be exposed. But from the first sequence, \( \rho_1 \) must be eliminated before \( \rho_2 \) can be exposed. Together, this implies that neither \( \rho_1 \) nor \( \rho_2 \) can be exposed without the other having already been eliminated — a contradiction. Therefore, our assumption must be false, and we conclude that \( \rho_1 = \rho_2 \). Hence, \( \leq \) is antisymmetric.

    \item \textbf{Transitivity:} Let \( \rho_1, \rho_2, \rho_3 \in R \) such that \( \rho_1 \leq \rho_2 \) and \( \rho_2 \leq \rho_3 \). We show that \( \rho_1 \leq \rho_3 \). By the definition of \( \leq \), either \( \rho_1 = \rho_2 \) or there exists a finite sequence of meta-rotations \( \rho_1 \prec \rho_{u} \prec \cdots \prec \rho_2 \), and similarly, either \( \rho_2 = \rho_3 \) or there exists a finite sequence \( \rho_2 \prec \rho_{v} \prec \cdots \prec \rho_3 \). If \( \rho_1 = \rho_2 \), then \( \rho_1 \leq \rho_3 \) follows directly from \( \rho_2 \leq \rho_3 \). If \( \rho_2 = \rho_3 \), then \( \rho_1 \leq \rho_3 \) follows from \( \rho_1 \leq \rho_2 \).

Otherwise, we can combine the two sequences of \( \prec \) relations to obtain:
\[
\rho_1 \prec \rho_{u} \prec \cdots \prec \rho_2 \prec \rho_{v} \prec \cdots \prec \rho_3,
\]
which is itself a finite sequence of meta-rotation eliminations from \( \rho_1 \) to \( \rho_3 \). Therefore, \( \rho_1 \leq \rho_3 \) by definition of \( \leq \), and so the relation is transitive.
\end{enumerate}
\end{proof}

\begin{definition}[\textbf{Closed subset}]
\label{def:closed-subset}
    A subset of \( (R, \leq) \) is said to be \textit{closed} if, for every \( \rho \) in the subset, all \( \rho' \in R \) such that \( \rho' \leq \rho \) are also contained in the subset.
\end{definition}

\noindent Finally, we present Lemma~\ref{lem:no-pair-in-two-metarotation}, which states that no pair \((s_i, p_j)\) belongs to more than one meta-rotation in \( I \). For the remainder of the paper, we denote the meta-rotation poset \( (R, \leq) \) of \( I \) by \( \Pi(I) \).

\begin{lemma}
\label{lem:no-pair-in-two-metarotation}
    Let $I$ be a given {\sc spa-s} instance. No pair $(s_i,p_j)$ can belong to two different meta-rotations in $I$.
\end{lemma}
\begin{proof}
Let $I$ be a given {\sc spa-s} instance. Suppose, for contradiction, that a pair \( (s_i, p_j) \) appears in two different meta-rotations \( \rho_1 \) and \( \rho_2 \), i.e., \( (s_i, p_j) \in \rho_1 \cap \rho_2 \) and \( \rho_1 \ne \rho_2 \). Since the meta-rotations are distinct, there exists at least one pair \( (s', p') \in \rho_1 \setminus \rho_2 \). We consider cases (A) and (B), depending on whether \( \rho_1 \) and \( \rho_2 \) are exposed in the same stable matching or in different ones.

\vspace{0.1cm}
\noindent \textbf{Case (A):} \( \rho_1 \) and \( \rho_2 \) are both exposed in the same stable matching \( M \).
\noindent Then, \( (s_i, p_j) \in M \). Eliminating \( \rho_2 \) from \( M \) yields a new stable matching \( M^* = M / \rho_2 \), where each student in \( \rho_2 \) is assigned to a less preferred project. So, \( s_i \) prefers \( p_j \) to \( M^*(s_i) \). Let \( M_L \) be the lecturer-optimal stable matching. Then either \( M^* = M_L \), or \( M^* \) dominates \( M_L \). In either case, it follows that \( s_i \) is assigned to different projects in \( M \) and \( M_L \). By Corollary~\ref{cor:each-stud-in-one-rotation}, any student who is assigned to different projects in \( M \) and \( M_L \) is involved in at most one exposed meta-rotation of \( M \). Since \( s_i \in \rho_2 \), and \( \rho_2 \) is exposed in \( M \), then \( s_i \) cannot also be in \( \rho_1 \), contradicting the assumption that \( (s_i, p_j) \in \rho_1 \cap \rho_2 \).

\vspace{0.2cm}
\noindent \textbf{Case (B):} Suppose \( \rho_1 \) and \( \rho_2 \) are exposed in different stable matchings. Let \( M_1 \) be a stable matching in which \( \rho_1 \) is exposed, and let \( M_2 \) be a stable matching in which \( \rho_2 \) is exposed. Recall that \( (s_i, p_j) \in \rho_1 \cap \rho_2 \), and \( (s', p') \in \rho_1 \setminus \rho_2 \). Since \( \rho_2 \) is exposed in \( M_2 \), it follows that \( M_2(s_i) = p_j \).
Moreover, $s'$ is assigned in $M_2$. Suppose that \( s' \) prefers \( p' \) to \( M_2(s') \). Then by Corollary~\ref{cor:s-in-rho-all-s-in-rho}, since both $(s_i,p_j)$ and $(s',p')$ are in $\rho_1$, then \( s_i \) also prefers \( p_j \) to \( M_2(s_i) \); however, this contradicts the fact that \( M_2(s_i) = p_j \). Hence, \( s' \) either prefers \( M_2(s') \) to \( p' \), or $M_2(s') = p'$. Let \( M_2(s') = p_x \), and let \( M^* \) be the stable matching obtained by eliminating \( \rho_2 \) from \( M_2 \). We consider subcases (B1) and (B2) depending on whether $(s',p_x) \in \rho_2$. 

\vspace{0.2cm}
\noindent \textbf{Case (B1):} \( (s', p_x) \in \rho_2 \). Since \( (s', p') \notin \rho_2 \), we have that \( p_x \neq p' \) and \( s' \) prefers \( p_x \) to \( p' \). After eliminating \( \rho_2 \), \( s_i \) is worse off in \( M^* \) than in $M_2$, i.e., \( s_i \) prefers \( p_j \) to \( M^*(s_i) \). Meanwhile, \( s' \) either becomes assigned to \( p' \) (that is, \( M^*(s') = p' \)), or \( s' \) prefers \( p_x \) to \( M^*(s') \), and prefers \( M^*(s') \) to \( p' \). We note that \( s' \) does not prefer \( p' \) to \( M^*(s') \), since by Lemma~\ref{lem:no-p-between-M(s_i)-and-sM(s_i)}, if \( p' \) lies between \( p_x \) and \( M^*(s') \) on the preference list of \( s' \), then \( (s', p') \) is not a stable pair. This means that \( (s', p') \) cannot be in \( \rho_1 \). Thus, \( s' \) does not prefer \( p' \) to \( M^*(s') \), while \( s_i \) prefers \( p_j \) to \( M^*(s_i) \). Thus, one student (namely $s_i$) in $\rho_1$ prefers their project in $\rho_1$ to their assignment in $M^*$, while another student (namely $s'$) does not, contradicting Corollary~\ref{cor:s-in-rho-all-s-in-rho}.

\vspace{0.2cm}
\noindent \textbf{Case (B2):} \( (s', p_x) \notin \rho_2 \). Then $s'$ remains assigned to $p_x$ in $M^*$, that is, \( M^*(s') = p_x \). Recall that either \( s' \) prefers \( p_x \) to \( p' \) or \( p_x = p' \). By Corollary~\ref{cor:s-in-rho-all-s-in-rho}, since $(s_i,p_j) \in \rho_1$ and $s_i$ prefers $p_j$ to $M^*(s_i)$ then $s'$ should prefer $p'$ to $M^*(s')$, a contradiction. 

\vspace{0.1cm}
\noindent Hence, no pair belongs to two different meta-rotations in $I$.
\end{proof}

\noindent We now present a nice structural relationship between the closed subsets of $\Pi(I)$ and the stable matchings of $I$.

\begin{theorem}
Let $I$ be a {\sc spa-s} instance. There is a one-to-one correspondence between the set of stable matchings in $I$ and the closed subsets of the meta-rotation poset $
\Pi({I})$ of $I$.
\end{theorem}

\begin{proof}
Let \( I \) be a given {\sc spa-s} instance, and let \( R \) denote the set of all meta-rotations in \( I \). First, we show that each closed subset of meta-rotations in \( \Pi(I) \) corresponds to exactly one stable matching of \( I \). Let \( A \subseteq R \) be a closed subset of \( \Pi(I) \). By definition, if a meta-rotation \( \rho \in A \), then all predecessors of \( \rho \) in \( \Pi(I) \) also belong to \( A \). Hence, it is possible to eliminate all meta-rotations in \( A \) in some order consistent with the partial order $\leq$, starting from the student-optimal stable matching. By Lemma~\ref{lem:M-rho-is-stable}, each such elimination step results in another stable matching of $I$, and the final matching obtained after eliminating all meta-rotations in $A$ is stable.

Suppose $A_1$ and $A_2$ are two distinct closed subsets of \( \Pi(I) \). Since \( A_1 \neq A_2 \), there exists at least one meta-rotation \( \rho \) that belongs to one of the subsets and not the other. Furthermore, since no two meta-rotation contains the same set of student-project pairs by Lemma~\ref{lem:no-pair-in-two-metarotation}, we would obtain two different stable matchings of $I$ when we eliminate the meta-rotations in $A_1$ and $A_2$. Therefore, eliminating each closed subset results in a unique stable matching.

\medskip
We now prove the converse: that each stable matching \( M \in \mathcal{M} \) corresponds to a unique closed subset of \( \Pi(I) \). Let \( A \subseteq \Pi(I) \) denote the set of meta-rotations that are eliminated, starting from the student-optimal stable matching \( M_s \), in order to obtain \( M \). This set must be closed; that is, if some meta-rotation \( \rho_2 \in A \) and \( \rho_1 \leq \rho_2 \) in \( \Pi(I) \), then \( \rho_1 \) must have been eliminated before \( \rho_2 \) could be exposed, and hence \( \rho_1 \in A \). It follows that \( A \) contains all predecessors of its elements and is therefore a closed subset. 

Now, consider two different stable matchings \( M, M' \in \mathcal{M} \). Then there exists a pair \( (s_i, p_j) \in M \setminus M' \). We prove that the sets of eliminated meta-rotations that yield \( M \) and \( M' \) differ. First, suppose \( M \) is the student-optimal matching \( M_s \). In this case, no meta-rotation is eliminated to obtain \( M \), but \( (s_i, p_j) \) must have been removed during the construction of \( M' \) by eliminating some meta-rotation \( \rho \). Thus, \( \rho \) is eliminated in the construction of \( M' \), but not \( M \).

Suppose \( M \neq M_s \). If \( (s_i, p_j) \) does not belong to \( M_s \), then \( (s_i, p_j) \) was introduced to \( M \) by eliminating some meta-rotation \( \rho \). By Lemma~\ref{lem:no-pair-in-two-metarotation}, each pair appears in at most one meta-rotation. Hence, \( s_i \) was assigned to \( p_j \) in \( M \) through the elimination of exactly one meta-rotation, namely \( \rho \). Since \( (s_i, p_j) \in M \setminus M' \), \( \rho \) must have been eliminated in constructing \( M \), but not in \( M' \). If \( (s_i, p_j) \) belongs to \( M_s \), then no meta-rotation involving \( (s_i, p_j) \) was eliminated in the construction of \( M \), but \( (s_i, p_j) \) must have been removed in the construction of \( M' \) by eliminating some meta-rotation \( \rho \). Hence, the sets of eliminated meta-rotations for \( M \) and \( M' \) differ. Thus, each stable matching corresponds to a unique closed subset of \( \Pi(I) \).
\end{proof}

\section{Conclusion}
In this paper we introduced the concept of meta-rotations in {\sc spa-s}, generalising the notions of rotations and meta-rotations from one-to-one and many-to-many models. 
We established a one-to-one correspondence between the set of stable matchings in an instance and the family of closed subsets of its meta-rotation poset \( \Pi(M) \), providing a compact characterisation of all stable matchings. This result has direct algorithmic implications, similar to those established for {\sc sm} and {\sc hr}: it enables the enumeration and counting of all stable matchings in {\sc spa-s}, and supports the design of algorithms for computing optimal matchings under various objectives, such as egalitarian and minimum-cost solutions. It also provides a foundation for studying the structural properties and computational complexity of various types of stable matchings in {\sc spa-s}.

A promising direction for future work is to develop a polyhedral characterisation of the set of stable matchings, by identifying inequalities whose feasible region exactly describes all stable matchings and proving that the corresponding polytope is integral. Such a formulation would enable new linear programming techniques for solving optimisation problems involving stable matchings in {\sc spa-s}. It could also serve as a foundation for proving that the polytope describing strongly stable and super-stable matchings in the {\sc spa-s} setting with ties in preferences \cite{olaosebikan2022super,olaosebikan2020student} are integral, thereby extending known integrality results for related models such as the Hospital–Residents problem with ties \cite{kunysz2016characterisation,hu2021characterization}.

%
%
\bibliographystyle{splncs04}
\bibliography{ref}
\end{document}